\documentclass[11pt,reqno]{article}
\usepackage{bm}
\usepackage{amsthm}
\usepackage{amsmath}
\usepackage{amssymb}
\usepackage[round]{natbib}
\usepackage{color}
\usepackage{authblk}
\usepackage{microtype}
\usepackage{graphicx}
\usepackage{float}
\usepackage{prodint}
\usepackage{paralist}
\usepackage{epstopdf}
\usepackage{geometry}
\DeclareMathSymbol{\shortminus}{\mathbin}{AMSa}{"39}
\newcommand{\rnc}{\renewcommand}
\newcommand{\nc}{\newcommand}
\newcommand{\mrm}{\mathrm}

\nc{\mb}{\mathbb}
\nc{\mc}{\mathcal}
\nc{\E}{\mb{E}}
\nc{\N}{\mb{N}}
\nc{\R}{\mb{R}}
\nc{\Q}{\mb{Q}}
\rnc{\P}{\mrm P}
\rnc{\d}{\mrm d}
\nc{\C}{\mc{C}}
\nc{\D}{\mc{D}}
\nc{\B}{\mc{B}}
\nc{\vbeta}{\bm \beta}
\nc{\vtheta}{\bm \theta}
\nc{\vX}{\bm X}
\nc{\vy}{\bm y}
\nc{\vU}{\bm U}
\nc{\vI}{\bm I}
\nc{\vE}{\bm E}
\nc{\ve}{\bm e}
\nc{\vV}{\bm V}
\nc{\vv}{\bm v}
\nc{\vS}{\bm S}
\nc{\vSigma}{\bm \Sigma}
\nc{\oPo}{\stackrel{\mrm p}{\rightarrow}}
\nc{\oWo}{\stackrel{w}{\rightarrow}}
\nc{\oDo}{\stackrel{d}{\longrightarrow}}
\nc{\eff}{\|F\|}

\def\E{{ E }}
\def\R{{ \mathbb{R} }}
\def\N{{ \mathbb{N} }}

\def\P{ P }

\def\E{ E }

\newcommand{\mbf}{\boldsymbol} % for bold symbols
\def\trans{{ ^\mathrm{\scriptscriptstyle T} }}

\newtheorem{lemma}{Lemma}
\newtheorem{assumption}{Assumption}
\newtheorem{corollary}{Corollary}
\newtheorem{theorem}{Theorem}

\newcommand\blfootnote[1]{%
  \begingroup
  \renewcommand\thefootnote{}\footnote{#1}%
  \addtocounter{footnote}{-1}%
  \endgroup
} % Fuer Fussnoten ohne Nummer

\begin{document}

\title{\Large \bf Permutation inference in factorial survival designs with the CASANOVA}
\author[1,$*$]{Marc Ditzhaus}
\author[2]{Arnold Janssen}
\author[1]{Markus Pauly}

\affil[1]{Department of Statistics, TU Dortmund University, Germany.}
\affil[2]{Mathematical Institute, Heinrich-Heine University Duesseldorf, Germany.}

\maketitle

\begin{abstract}
\blfootnote{${}^*$ e-mail: marc.ditzhaus@tu-dortmund.de}
 We propose inference procedures for general nonparametric factorial survival designs with possibly right-censored data. Similar to additive Aalen models, null hypotheses are formulated in terms of cumulative hazards. Thereby, deviations are measured in terms of quadratic forms in Nelson--Aalen-type integrals. Different to existing approaches this allows to work without restrictive model assumptions as proportional hazards. In particular, crossing survival or hazard curves can be detected without a significant loss of power. For a distribution-free application of the method, a permutation strategy is suggested. The resulting procedures' asymptotic validity as well as their consistency are proven and their small sample performances are analyzed in extensive simulations. Their applicability is finally illustrated by analyzing an oncology data set.
\end{abstract}


%\vfill
%\vfill

\section{Introduction}\label{sec:intro}
\cite{kristiansen2012} reviewed 175 studies with time to event endpoints published 
in five renowned journals. In 47$\%$ of these studies crossing survival curves were present. The alarming observation of his review was: ``{Among 
	studies with survival curve crossings, Cox regression was performed in 66\% and log-rank-test in 70\% of the studies.}'' Under the assumption of proportional hazards the log-rank test and the Cox regression are indeed very powerful tools. Otherwise, however, log-rank tests significantly loose power and Cox regressions cannot be interpreted appropriately, in particular, when the survival curves cross.  Thus, as 
stated by \cite{bouliotis_billingham_2011}: ``{There is a need in the clinical community to clarify methods that are appropriate when survival curves 
	cross.}''
%For example, 
This especially holds in oncology, where crossing survival curves are frequently observed \citep{gahrton2013,weedaETAL2013fibrolamellar,smithETAL2014} due to a potentially delayed treatment effect of immunotherapy  \citep{mickChen2015}.
% effect combined with the long-term survival of immunotherapy-treated patients

In the two-sample setting, various inferences methods have been designed to detect nonproportional hazard alternatives and, in particular, crossing curves. We refer to \cite{liETAL_2015_comp} for a detailed review and to \cite{LiuDitzhausXu_2020_ABC} for a recent proposal. Thereof, a tempting approach are extensions of so-called weighted log-rank tests \citep{TaroneWare1977,Gill1980,ABGK,BathkeETAL2009}, for which the power is optimized for certain nonproportional hazard alternatives by adding a corresponding weight function. For example, the weights of \cite{HarringtonFleming1982} can be used for late treatment effects, as also recalled by \cite{fine2007} and \cite{SuZhu_2018}. Such weights probably would have helped  \cite{jacobsETAL2016} to confirm their initial assumption %meet their primary endpoint 
in an ovarian cancer screening trail. Instead they used the log-rank test and stated: ``{The  main  limitation  of  this  trial  was  our  failure  to  anticipate  the  late  effect  of  screening  in  our  statistical  design. Had we done so, the weighted log-rank test could have  been  planned  in  line  with  many  other  large  cancer  screening  trials.}'' This quote illustrates the problem of the log-rank test and its weighted versions: they are designed for specific alternatives and prior knowledge is needed to choose the optimal weight. To overcome this selection problem, \cite{BrendelETAL2014} introduced a combination approach of several weights leading to procedures with broader power functions. Recently, their approach was revisited and simplified leading to computationally more efficient test versions
\citep{ditzhaus:pauly:2019,ditzhaus:friedrich:2018}. 

It is the aim of the present paper to extend their idea %weighted log-rank test combination approaches 
to multiple samples and more general factorial designs. This way we not only address the $k$-sample problem with crossing survival curves for which only %the number of relevant literature shrinks significantly and there are just 
a handful of relevant methods  exist \citep{BathkeETAL2009,liuYin_2017, chen2016, gorfineETAL_2019} but also 
develop methods that fully %Moreover, beyond $k$-sample problems, methods that 
exploit the structure of factorial survival designs. % are of even greater interest. 
Such methods shall not only allow the detection of main factor effects (e.g. treatment or gender) but also infer potential interaction effects  \citep{gissi:1990,cassidy:etal:2008,green2012,kurz:etal:2015} as, e.g., also stated by \cite{lubsen:pocock:1994}: ``{it is desirable for reports of factorial trials to include estimates of the interaction between the treatments}''. In contrast to Cox, Aalen, or Cox-Aalen regression models \citep{cox1972regression,scheike2002additive,scheike2003extensions}, there is no need 
to introduce multiple dummy variable for nominal factors (e.g. different treatments) in the factorial design set-up, which is even favorable in uncensored situations \citep{green2002factorial,green2012}. In the context of survival data, there are just a few nonparametric methods accounting for factorial designs: the approaches of \cite{akritas:brunner:1997}, which require a strong assumption on the underlying censoring distribution that is often too strict from a practical point of view, 
%requiring a, which is hard to justify in practice, 
and the procedure of  \cite{dobler:pauly:2019}. However, the latter formulates null hypotheses in terms of certain concordance effects that restricts their analysis to a pre-specified time range $[0,\tau]$. Moreover, both approaches are not flexibly adoptable to detect certain crossing structures and have not yet been fully implemented into statistical software. 
%For al these reasons both approaches have not yet found their way into statistical practice.

In contrast, we follow the spirit of the additive model of \cite{aalen:1980} and  formulate our null hypotheses by means of cumulative hazard functions. % To this end, we extend the weighted log-rank combination approaches  of  \cite{BrendelETAL2014,ditzhaus:friedrich:2018,ditzhaus:pauly:2019} for the two-sample setting to  general factorial designs. 
Inspired by \cite{neuhaus:1993} and \cite{JanssenMayer2001}, we derive critical values by means of a permutation procedure. %which finally leads to a testing procedure that not only fulfills the desired asymptotic properties as being exact under the null hypothesis 
%but also performs satisfactorily for small sample sizes.
%(1) without any restrictive assumption on the censoring distribution or the time range, (2) with reasonable power under the Cox model as well as for crossing survival curve scenarios, (3) for general factorial designs allowing the study of main and interaction effects, (4). 
Under exchangeable data, e.g., when the survival and the censoring distributions are the same over all groups, respectively, the permutation strategy leads to a finitely exact test under the null hypothesis. Moreover, for non exchangeable situations, this exactness can be preserved at least asymptotically by following the idea of permuting studentized statistics. This desirable property of studentized permutation tests was mainly explored for testing means and other functionals in the two-sample case \citep{janssen:1997studentized,janssenPauls2003,pauly:2011:discussion}. It was later extended to one-way layouts by \cite{chung2013exact} and finally reached its full potential under general factorial designs \citep{pauly2015asymptotic, friedrich2017permuting, smaga2017diagonal, harrar2019comparison,ditzhaus2019qanova}.
In our more complex survival setting, the weighted combination approach paired with the permutation strategy leads to a test 
\begin{enumerate}[(a)]
	\item without any restrictive assumption on the censoring distribution or the time range, 
	\item with reasonable power under proportional hazards as well as for crossing survival curve scenarios, 
	\item for general factorial designs allowing the study of main and interaction effects, 
	\item being asymptotically valid with satisfactory small sample size performance.
\end{enumerate}

In Section \ref{sec:setup}, we introduce the survival model and the null hypotheses formulated in
terms of cumulative hazard rate functions. To test for certain main or interaction effects, we propose respective
Wald-type statistics based on weighted Nelson--Aalen-type integrals, see Section~\ref{sec:asy_properties}. We prove their asymptotic
validity and derive their power behavior under local alternatives. Motivated by the latter, we suggest to combine different weight functions into a joint Wald-type statistic to obtain a powerful method for various alternatives simultaneously, e.g. proportional hazards and crossing survival curves. Respective permutation versions of them, promising a better finite sample performance, are shown to be asymptotic exact in Section~\ref{sec:permutation}. A simulation study presented in Section~\ref{sec:simu} reveal an actual improvement when using the permutation approach and show that the combination strategy actually result in a powerful test for proportional hazards as well as crossing curves alternatives. Finally, the tests' applicability are illustrated by analyzing an oncology data set in Section~\ref{sec:real_data}.

\section{The set-up}\label{sec:setup}
%\subsection{Notation} 
Our general survival model is given by mutually independent positive random variables 
\begin{align}\label{eqn:model}
	T_{ji}\sim F_j,\quad C_{ji}\sim G_{j} \quad (j=1,\ldots,k;\,i=1,\ldots,n_j),
\end{align}
where $T_{ji}$ is the actual survival time with continuous distribution function $F_j$  and $C_{ji}$ denotes the corresponding right-censoring time with continuous distribution function $G_j$. This set-up allows the consideration of simple one-way but also of higher-way layouts. For illustration, let us consider a two-way design with factors $B$ (having $b$ levels) and $C$ (possessing $c$ levels). In this scenario, we set $k=b\cdot c$ and split up the group index $j$ into $j=(j_B,j_C)$ for $j_B=1,\ldots,b$ and $j_C=1,\ldots,c$. More complex designs, e.g. hierarchical designs with nested factors, can be incorporated into this framework as well. We  refer to \cite{dobler:pauly:2019} for more details. 

Based on the observation time  $X_{ji}=\min(T_{ji},C_{ji})$ and its censoring status $\delta_{ji}=I\{X_{ji}=T_{ji}\}$, where $I(\cdot)$ denotes the indicator function, we like to infer hypotheses formulated in terms of the cumulative hazard rate functions $A_j(t) = \int_0^t (1-F_j)^{-1}F_j$ $(t\geq 0$):
\begin{align}\label{eqn:null}
	\mathcal H_0:  \mbf{H} \mbf{A} = \mbf{0}_d, \quad\mbf{A}=(A_{1},\ldots,A_k)\trans,
\end{align}
where $\mbf{H}\in \R^{d\times k}$ is a contrast matrix, i.e. $\mbf{H}\mbf{1_k} = \mbf{0}_d$, and $\mbf{0}_d$ as well as $\mbf{1}_d$ are vectors in $\R^d$ consisting of $0$'s and $1$'s only. Here and subsequently, we use the following standard matrix notation: $B\trans$ is the transpose and $\mbf{B}^+$ is the Moore--Penrose inverse of a matrix $\mbf{B}$. The contrast matrix in \eqref{eqn:null} is chosen in regard to the underlying question of interest. 
%\subsection{Examples of contrast matrices}\label{sec:contrast_matrices}
%For the readers' convenience, we discuss some scenarios covered by our framework and explain how to choose the contrast matrix $\mbf H$ in \eqref{eqn:null}. 
For example in a one-way layout, the null hypothesis $$\mathcal H_0: \{A_{1}=\ldots=A_{k} \} = \{\mbf{P}_k \mbf{A} = \mbf{0}_k\}$$ of no group effect can be expressed in terms of the  contrast matrix $\mbf{P}_k = \mbf{I}_k - (\mbf{J}_k/k)$, where $\mbf{I}_{k}$ is the $k\times k$-dimensional unity matrix and $\mbf{J}_k = \mbf{1}_k\trans\mbf{1}_k\in \R^{k\times k}$ consists of 1 only. 

Switching to a two-way layout ($k=bc$) with the factors $B$ (having $b$ levels) and $C$ (possessing $c$ levels), the relevant contrast matrices are $\mbf{H}_{B}=  \mbf{P}_b \otimes (\mbf{J}_c/c)$, $\mbf{H}_C = (\mbf{J}_b/b) \otimes  \mbf{P}_c$ and $\mbf{H}_{BC} = \mbf{P}_b \otimes  \mbf{P}_c$, where  $\otimes$ is the Kronecker product. They can be used to check the null hypotheses 
\begin{itemize}
	\item \textit{No main effect B:} $\{\mbf{H}_{B}\mbf{A} = \mbf{0}_{k} \} = \{ \bar{A}_{1\cdot} = \ldots = \bar{A}_{b\cdot}\}$
	
	\item \textit{No main effect C:} $\{\mbf{H}_{C}\mbf{A} = \mbf{0}_{k} \} = \{ \bar{A}_{\cdot1} = \ldots = \bar{A}_{\cdot c}\}$
	
	\item \textit{No interaction effect:} $\{\mbf{H}_{BC}\mbf{A} = \mbf{0}_{k} \} = \{ \bar{A}_{\cdot \cdot}- \bar{A}_{\cdot j_C} - \bar{A}_{j_B\cdot} + {A}_{j_Bj_C}   = 0\}$
\end{itemize}
Here, $\bar{A}_{j_B \cdot}$, $\bar{A}_{\cdot j_C}$ and $\bar{A}_{\cdot \cdot}$ are the means over the dotted indices. In case of existing hazard rates $\alpha_j(t)  = dA_j(t)/dt$ $(t\geq 0)$ and having the additive Aalen model in mind, we can rewrite these null hypotheses in a more lucid way by decomposing the hazard rate $\alpha_j=\alpha_{j_Bj_C}$ into 
\begin{align*}
	\alpha_{j_Bj_C}(t) = \alpha_{0}(t) + \beta_{j_B}(t) + \gamma_{j_C}(t) + (\beta\gamma)_{j_Bj_C}(t) 
\end{align*}
with side conditions $\sum_{j_B} \beta_{j_B} = \sum_{j_C} \gamma_{j_C} =  \sum_{j_B} (\beta\gamma)_{j_Bj_C} = \sum_{j_C} (\beta\gamma)_{j_Bj_C} = 0$. Then we can rewrite $\{\mbf{H}_{C}\mbf{A} = \mbf{0}_k \} =\{ \gamma_{j_C} = 0\text{ for all }j_C\}$ or 
$$\{\mbf{H}_{BC}\mbf{A} = \mbf{0}_{k} \} =\{ (\beta\gamma)_{j_Bj_C} = 0\text{ for all }j_B,j_C\}$$
for the interaction hypothesis.
For higher-way layouts and more complex designs, such as nested settings, we refer the reader to \cite{pauly2015asymptotic} and \cite{dobler:pauly:2019}.

%In Section \ref{sec:contrast_matrices}, we exemplify how to choose $\mbf H$ for testing main and interaction effects in a two-factorial design. 

As in analysis-of-variance settings \citep{brunner:dette:munk:1997,pauly2015asymptotic,smaga2017diagonal}, it is preferable to work with the projection matrix $\mbf{T}=\mbf{H}\trans(\mbf{H}\mbf{H}\trans)^+ \mbf{H}$ over $\mbf{H}$ itself. Beside from being unique, $\mbf{T}$ is symmetric and idempotent, and describes the same null hypothesis as $\mbf H$ does. We will therefore work with $\mbf{T}$ when formulating our testing procedure. In addition, we need the usual counting process notation. Thus, let $N_j(t)=\sum_{i=1}^{n_j}I\{X_{ji}\leq t, \delta_{ji}=1\} $ be the number of observed events within group $j$ until time $t$ and introduce $Y_j(t)=\sum_{i=1}^{n_j}I\{X_{ji}\geq t\}$, the number of individuals being at risk just before $t$ in the same group. These processes allow us to define the Nelson--Aalen estimator for $A_j$ given by
\begin{align*}
	\widehat A_{j}(t)=\int_{0}^t \frac{I\{Y_j(s)>0\}}{Y_j(s)}\,\mathrm{ d }N_j(s)\quad (j=1,2;\,t\geq 0).
\end{align*}
For our purposes, the Kaplan--Meier estimator is only required for all $n=\sum_{j=1}^kn_j$  observations without distinguishing between the groups. This so-called pooled Kaplan--Meier estimator $\widehat F$ can be expressed in terms of the pooled counting processes $N= \sum_{j=1}^kN_j$ and $Y=\sum_{j=1}^kY_j$ by
\begin{align*}
	1-\widehat F(t)=\prod_{(j,i):X_{ji}\leq t}\Bigl( 1- \frac{\delta_{ji}}{Y(X_{ji})} \Bigr)=\prod_{(j,i):X_{ji}\leq t}\Bigl( 1- \frac{\Delta N(X_{ji})}{Y(X_{ji})} \Bigr)\quad (t\geq 0),
\end{align*}
where $\Delta N(t)=N(t) - N(t-)$ denotes the increment of the counting process $N$ at time $t$. In the same way, % Replacing the group-specific processes by their pooled counterparts, we can introduce also 
$\widehat A(t)=\int_0^t I\{Y>0\}/Y \,\mathrm{ d }N$ $\;(t\geq 0)$ denotes the pooled Nelson--Aalen estimator.

\section{Asymptotic results}\label{sec:asy_properties}
\subsection{Wald-type test}
Throughout, we assume non-vanishing groups $n_j/n\to \kappa_j\in(0,1)$ as $\min(n_j:j=1,\ldots,k)\to 0$. Moreover, we exclude the trivial case of purely censored observations in any of the groups by assuming that $F_j(t)>0$ and $G_j(t)<1$ for all $j=1,\ldots,k$ and some $t>0$. 

Weighted log-rank statistics \citep{FlemingETAL1987,ABGK} of the form
\begin{align*}
	\Big(\frac{n}{n_1n_2}\Big)^{1/2} \int_0^\infty \widetilde w\{ \widehat F_n(t-)\}\frac{Y_1(t)Y_2(t)}{Y(t)} \Big\{\mathrm{ d }\widehat A_1(t) - \mathrm{ d }\widehat A_2(t)\Big\},
\end{align*}
will later build the fundament of our new test statistics. Here, $t\mapsto \widehat F_n(t-)$ is the left continuous version of $\widehat F_n$ and $\widetilde w$ is a weight function taken from the space $\mathcal W$ consisting of all continuous functions $\widetilde w:[0,1] \to \R$ of bounded variations with $\widetilde w(t)\neq 0$ for some $t\in[0,1]$.  \cite{FlemingHarrington} considered a subclass of these weights having the shape $\widetilde w(t)=t^r(1-t)^g$ $(r,g\in\N_0)$, see their Definition 7.2.1. Setting $r=g=0$ we obtain the log-rank test \citep{Mantel1966,PetoPeto1972} and the choice $(r,g)=(1,0)$ leads to the Prentice--Wilcoxon test. In general, these weights can be used to prioritize (mid-)late, (mid-)early or central times by choosing $r,g$ appropriately. For our purposes, weights, e.g. $\widetilde w(t)=1-2t$, intersecting the $x$-axis are of special interest because they are designed for crossing hazard alternatives. Having all this weights at hand, the question arises: which weight should be chosen? We address this question in detail in the next two sections, but first we introduce the relevant components of the newly developed test statistic. For the sake of a clear and simple presentation, we restrict here to polynomial weights $\widetilde w$ covering the main relevant cases. However, more general weight functions can be treated analogously as discussed in the supplementary material. 
\begin{assumption}\label{ass:w}
	Let $\widetilde w\in \mathcal W$ be a nontrivial polynomial.
\end{assumption}

First, we extend the weighted log-rank integrand to the present situation of multiple subgroups
\begin{align}\label{eqn:def_wn}
	w_n(t) = \widetilde w\{ \widehat F_n(t-)\}\frac{ Y_1(t)\ldots Y_k(t)}{nY(t)^{k-1}}\quad (t\geq 0)
\end{align}
and then define the Nelson--Aalen-type integral over these new integrands
\begin{align*}
	Z_{nj}(w_n)= n^{1/2}\int_0^\infty w_n(t) \,\mathrm{ d }\widehat A_j(t).
\end{align*}
Using the standard martingale approach \citep{Gill1980,ABGK} we can show asymptotic normality for a centred version of this integral
\begin{align*}
	\widetilde Z_{nj}(w_n) = n^{1/2}\int_0^\infty w_n(t) \,\Big\{\mathrm{ d }\widehat A_j(t) - \mathrm{ d } A_j(t) \Big\}.
\end{align*} 
\begin{theorem}\label{theo:null_unc}
	Under Assumption \ref{ass:w}, ${\widetilde Z}_{nj}(w_n)$ converges in distribution to $N\{0,\sigma_j^2(w)\}$ for $\sigma_j^2(w)>0$ as $n\to\infty$, where $\sigma_j^2(w)$ can be consistently estimated by
	\begin{align*}
		\widehat \sigma_j^2(w_n) = n\int_0^\infty  \frac{w_{n}(t)^2}{Y_j(t)} \,\mathrm{ d }\widehat A_j(t).
	\end{align*}
\end{theorem}
Now, we are able to formulate a Wald-type statistic for testing $\mathcal H_0: \mbf{T}\mbf{A} = \mbf{0}_k$ given by
\begin{align}\label{eqn:def_WTS_one+weight}
	S_n(w_n) = [\mbf{T}\mbf{Z}(w_n)]\trans(\mbf{T}\mbf{\widehat\Sigma}(w_n)\mbf{T})^+\mbf{T}\mbf{Z}(w_n), \quad \mbf{\widehat \Sigma}(w_n) = \text{diag}\{\widehat \sigma^2_1(w_n),\ldots,\sigma^2_1(w_n)\}
\end{align}
and conclude from Theorem 9.2.2 of \cite{rao:mitra:1971} and Theorem \ref{theo:null_unc}:
\begin{corollary}\label{cor:WTS_null_one+weight}
	Under Assumption \ref{ass:w} and $\mathcal H_0:\mbf{T}\mbf{A}=\mbf{0}_k$, $S_n(w_n)$ converges in distribution to a $\chi^2_f$-distribution with $f=\text{rank}(\mbf{T})$ degrees of freedom as $n\to\infty$.
\end{corollary}
To motivate the combination approach of several weights, we study the asymptotic power behavior of $S_n(w_n)$ under local alternatives.

\subsection{Local alternatives}\label{sec:local}
%To study the power behavior under local alternatives,
To this end, we start with a fixed null setting given by a vector $\mbf{A}=(A_1,\ldots,A_k)^T$ with $\mbf{TA}=\mbf{0}_k$ and corresponding hazard rates $\alpha_j(t)= \,\mathrm{ d }A_j(t)/\,\mathrm{ d }t$ $(t\geq 0)$. Disturbing them as follows, we get a local alternative $(A_{1n},\ldots,A_{nk})$ tending with a rate of $n^{-1/2}$ to the null setting $\mbf{A}$:
\begin{align}\label{eqn:local_alternatives}
	\frac{\alpha_{nj}(t)}{\alpha_j(t)} = 1 + n^{-1/2} \gamma_j(t) \quad (j=1,\ldots,k;t\geq 0),
\end{align}
where the right hand side is non-negative and $\int_0^t \gamma_j(u) \alpha_j(u) \,\mathrm{ d }u \in \R$ for all $t\geq 0$ fulfilling $F_j(t)<1$. To simplify the situation, we may restrict to perturbations in the same direction but with possibly different strengths 
\begin{align}\label{eqn:def_gammaj}
	\gamma_j(t) = \theta_j \gamma\{F_0(t)\}.
\end{align}
Here, $F_0$ is the limit function of the pooled-Kaplan--Meier estimator, see the supplement for its concrete shape. %Under $H_0:A_1=\ldots=A_k$ this limit function reduces to $F_0=F_1$.
Moreover, we introduce $y_j=\kappa_j(1-G_j)(1-F_j)$ $(j=1,\ldots,k)$, which is the limit of $n^{-1}Y_j$.
\begin{theorem}\label{theo:local_alternatives}
	Under Assumption \ref{ass:w} and \eqref{eqn:local_alternatives}, $S_n(w_n)$ converges to a non-central $\chi_f^2(\delta)$-distribution with $f=\text{rank}(\mbf{T})$ and $\delta = (\mbf{T}\mbf{\mu})\trans(\mbf{T}\mbf{\Sigma}\mbf{T})^+\mbf{T}\mbf{\mu}$, where $\mbf\mu=(\mu_1,\ldots,\mu_k)\trans$ and
	\begin{align*}
		\mu_j = \int_0^\infty \widetilde w\{F_0(t)\} \frac{y_1(t)\ldots y_k(t)}{y(t)^{k-1}} \gamma_j(t) \,\mathrm{ d }A_j(t), \quad y=\sum_{j=1}^ky_j.
	\end{align*}
\end{theorem}
The effect of the weight function on the power under certain local alternatives can be illustrated best for the $k$-sample setting under  \eqref{eqn:def_gammaj}. In this case the non-centrality parameter simplifies to
\begin{align*}
	\delta = \Bigl[ \int_0^\infty \widetilde w\{F_0(t)\} \gamma\{F_0(t)\} \frac{y_1(t)\ldots y_k(t)}{y(t)^{k-1}} \,\mathrm{ d }A_1(t) \Bigr]^2  (\mbf{T}{\mbf\theta})\trans(\mbf{T}\mbf{\Sigma} \mbf{T})^+\mbf{T}{\mbf\theta},
\end{align*}
where $\mbf{\theta} = (\theta_1,\ldots,\theta_k)\trans$ and $\mbf{\Sigma}= \text{diag}\{\sigma^2_1(\widetilde w),\ldots,\sigma^2(\widetilde w)\}$. Consequently, choosing $\widetilde w$ as a multitude of $\gamma$ leads to the highest value for $\delta$ and, consequently, to the highest power of $S_n(w)$. However, the direction $\gamma$ of the departure from the null hypothesis is unknown and again the question arises: how to choose $\widetilde w$? The task of finding the optimal $\widetilde w$ is impossible. The most popular choice is the log-rank test $( \widetilde w \equiv 1)$ which, however, lacks to detect crossing hazard departures. To compensate for that, we follow \cite{ditzhaus:friedrich:2018} and suggest to combine the log-rank weight and a weight for crossing hazard alternatives, e.g. $\widetilde w(x) =1-2x$ $(0\leq x \leq 1)$, into a joint Wald-type statistic. In general, the new approach is not restricted to these two weights and even more than two weights can be combined.

\subsection{Combination of different weights}\label{sec:comb_weights}
Let us start with an arbitrary number of pre-chosen weights $\widetilde w_1,\ldots,\widetilde w_m$ corresponding to alternatives of interest, e.g. proportional, late, early or crossing hazards. Moreover, let $w_{n1},\ldots,w_{nk}$ be the corresponding integrands of the form \eqref{eqn:def_wn} for the Nelson--Aalen-type integrals. To exclude redundant cases, as too similar weights or even $\widetilde w_r = \widetilde w_{r'}$ for $r\neq r'$, we follow the suggestion of  \cite{ditzhaus:friedrich:2018} and \cite{ditzhaus:pauly:2019} and restrict to weights fulfilling
\begin{assumption}\label{ass:lin_ind}
	Let $\widetilde w_1,\ldots, \widetilde w_m\in \mathcal W$ be linearly independent, nontrivial polynomials.
\end{assumption}
The basic idea is now to combine $\mbf{Z}_n(w_{n1}),\ldots,\mbf{Z}_n(w_{nm})$ into one joint Wald-type statistic. For this purpose, we introduce the block diagonal matrix %the $m$-times direct sum $\mbf{T}^{(m)}$ of $\mbf{T}$ given by 
$\mbf{T}^{(m)} = %\bigoplus_{j=1}^k \mbf{T} = 
\text{diag}(\mbf{T},\ldots,\mbf{T})\in\R^{km \times km}$. Since $Z_{nj}(w_{nr})$ and $Z_{nj}(w_{nr'})$ are highly dependent, the vectors $\mbf{Z}_n(w_{nr})$ and $\mbf{Z}_n(w_{nr'})$ are so as well. Thus, the covariance matrix estimator required for the joint Wald-type statistic is not a simple diagonal matrix as in \eqref{eqn:def_WTS_one+weight}. In fact, the updated estimator has a block matrix representation $\mbf{\widehat \Sigma} = ( \mbf{\widehat{\Sigma}}^{(rr')})_{r,r'=1,\ldots,m}$, where each submatrix $\mbf{\widehat\Sigma}^{(rr')} = \text{diag}(\widehat\sigma^{2,(rr')}_1, \ldots,\widehat\sigma^{2,(rr')}_k)$ is a diagonal matrix with entires
\begin{align*}
	\widehat\sigma^{2,(rr')}_j = n\int_0^\infty  \frac{w_{nr}(t)w_{nr'}(t)}{Y_j(t)} \,\mathrm{ d }\widehat A_j(t).
\end{align*}
To sum up, we obtain the following updated Wald-type statistic:
\begin{align}\label{eqn:Sn_w1...wm}
	S_n = (\mbf{T}^{(m)}\mbf{Z}_n)\trans(\mbf{T}^{(m)}\mbf{\widehat \Sigma}\mbf{T}^{(m)})^+\mbf{T}^{(m)}\mbf{Z}_n,\quad \mbf{Z}_n = \{\mbf{Z}_{n}(w_{n1})\trans, \ldots, \mbf{Z}_{n}(w_{nk})\trans\}\trans.
\end{align}
\begin{theorem}\label{theo:several_weights}
	Under Assumption \ref{ass:lin_ind} and $\mathcal H_0:\mbf{T}\mbf{A}=\mbf{0}_k$, $S_n$ tends to a $\chi^2_f$-distribution with $f=m\cdot\text{rank}({T})$ degrees of freedom as $n\to\infty$. 
\end{theorem}
By Theorem \ref{theo:several_weights}, an asymptotically exact test $\phi_n = I\{S_n> \chi^2_{f,\alpha}\}$, i.e., $\E_{\mathcal H_0}(\phi_n)\to \alpha$ as $n\to\infty$, is derived by comparing the joint Wald-type statistic $S_n$ with the $(1-\alpha)$-quantile $\chi^2_{f,\alpha}$ of the $\chi^2_f$-distribution. However, simulation results from Section \ref{sec:simu} reveal a very conservative behavior of $\phi_n$ under small sample sizes. To tackle this problem, we suggest a permutation strategy leading to a better finite sample performance as can be seen in Section \ref{sec:simu}.

\section{Permutation test}\label{sec:permutation}
Resampling methods and, in particular, permutation procedures are well accepted tools to %approximate the critical values leading usually to a better 
improve the finite sample performance of asymptotic tests. The advantage of permuting over other resampling methods, e.g., bootstrap procedures, is the finite sample exactness of the test under exchangeable data, i.e. under the restrictive null hypothesis $\mathcal{\widetilde H}_0: A_1=\ldots=A_k,G_1=\ldots=G_k$ in our scenario. At the same time, 
the asymptotic exactness of the test beyond exchangeability can often be transferred to its permutation counterpart when working with studentized statistics, as the present joint Wald-type statistic. That is why we promote the following permutation strategy for our setting:  %the underlying inference problem. 
To obtain a permutation sample $\{(X_{ji}^\pi,\delta_{ji}^\pi):j=1,\ldots,k;i=1,\ldots,n_j\}$, we randomly interchange the group memberships of the observation pairs $(X_{ji},\delta_{ji})$. 
%That is, we simply draw without replacement from the pooled observation $\{(X_{ji},\delta_{ji}):j=1,\ldots,k;i=1,\ldots,n_j\}$.
With this we calculate the permutation version of the joint Wald-type statistic $S_n^\pi = S_n((X_{ji}^\pi,\delta_{ji}^\pi)_{j,i})$.
\begin{theorem}\label{theo:perm} 
	Under Assumption \ref{ass:lin_ind}, the permutation counterpart $S_n^\pi$ of $S_n$ always asymptotically mimics its null distribution, i.e. under $\mathcal H_0$ as well as under fixed and local alternatives \eqref{eqn:local_alternatives} we have as $n\to\infty$
	\begin{align*}
		\sup_{t\in[0,\infty)}  | \P \{ S_n^\pi \leq t \mid (X_{ji},\delta_{ji})_{ji} \} - \chi^2_{m\cdot\text{rank}(\mbf{T})}(-\infty,x] | \xrightarrow{p} 0.
	\end{align*}
\end{theorem}
The theorem allows to compare the joint Wald-type statistic $S_n$ with the $(1-\alpha)$-quantile $c_{n,\alpha}^\pi$ of $t\mapsto \P \{ S_n^\pi \leq t \mid (X_{ji},\delta_{ji})_{ji} \}$ (instead of the asymptotic $\chi^2_f$-quantile). %$\chi_{f,\alpha}$ 
This results in the permutation test $\phi_n^\pi = I\{S_n> c_{n,\alpha}^\pi\}$, which has the same asymptotic power and type-1 error behavior under $\mathcal H_0$ as well as under fixed and local alternatives \citep[Lemma 1]{janssenPauls2003}.

\section{Simulations}\label{sec:simu}
\begin{figure}
	\centering
	%graphics were generated in the folder !plots
	\begin{minipage}{0.47\textwidth}
		\includegraphics[width=\textwidth]{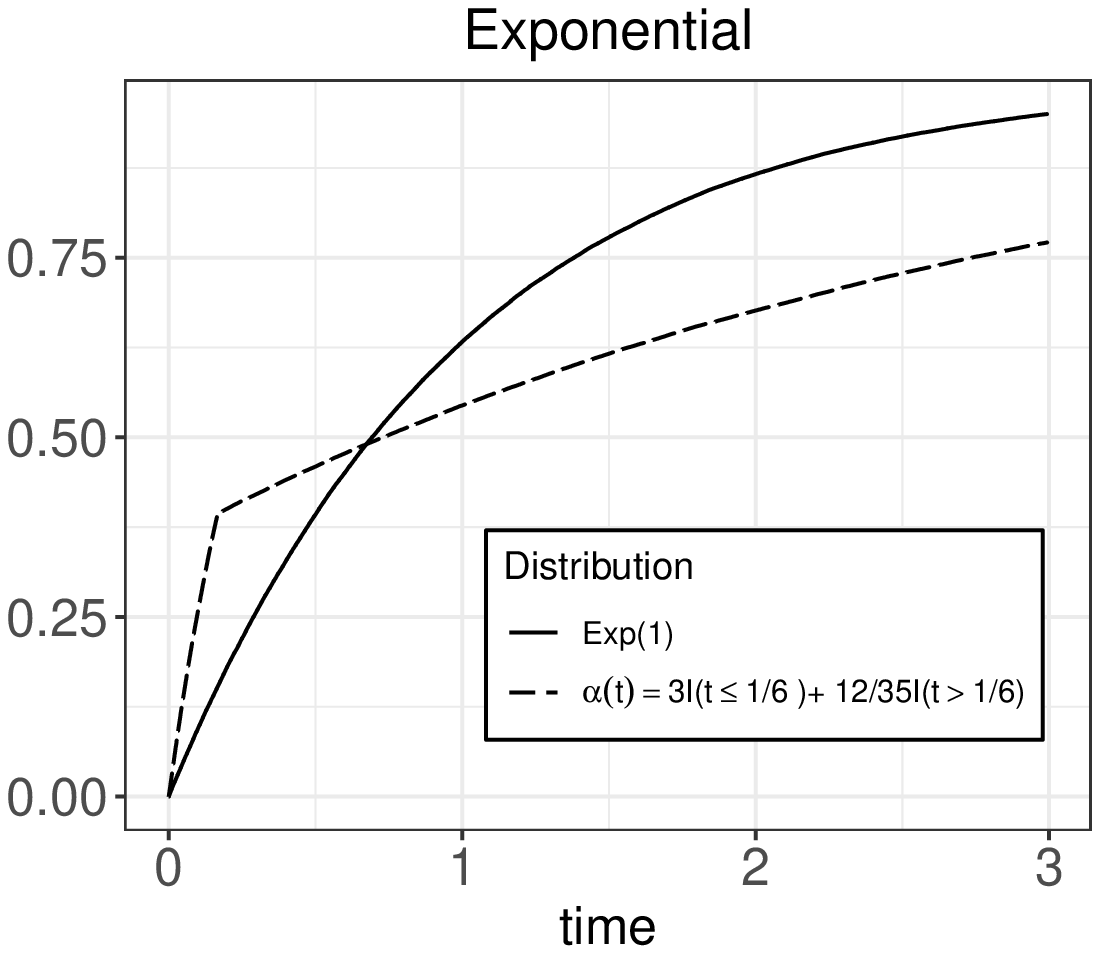}
	\end{minipage}
	\begin{minipage}{0.47\textwidth}
		\includegraphics[width=\textwidth]{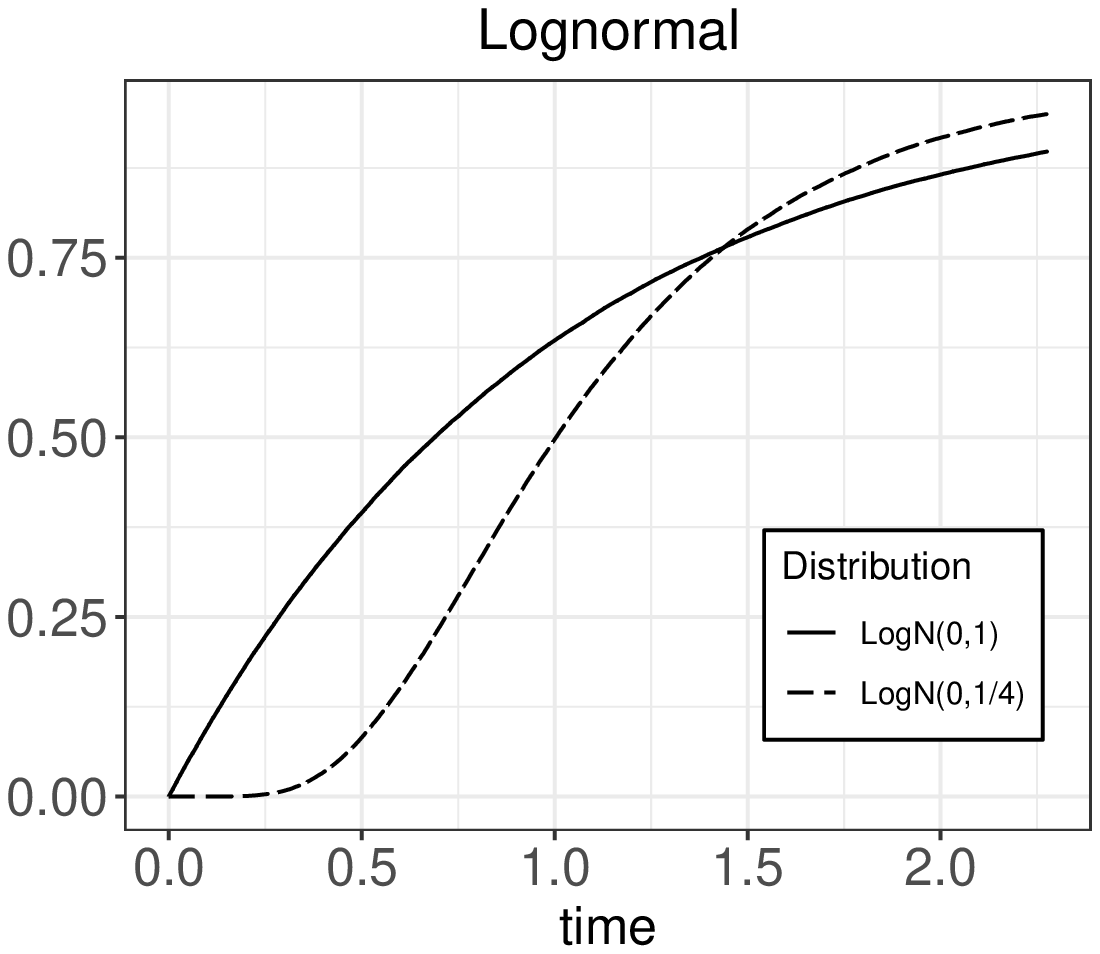}
	\end{minipage}
	\caption{Crossing curve alternatives}\label{fig:cross}
\end{figure}

In addition to the asymptotic findings, we conduct a simulation study to examine the small sample performance of the tests. In regard to our data example, we consider a $2\times 3$-design leading to $k=6$ subgroups and perform the tests for no main effect $\mathcal H_0:\{\mbf{H}_B \mbf{A} = \mbf{0}_k\}$ and for no interaction effect $\mathcal H_0:\{\mbf{H}_{BC}\mbf{A} = \mbf{0}_k\}$. Some additional simulation results for the one-way layout with $k=6$ are deferred to the supplement. For all these testing problems, we combine the classical log-rank weight $w_1(t)=1$ and a weight $w_2(t) = 1-2t$ $(0 \leq t\leq 1)$ for crossing survival curve alternatives. Under the null hypotheses, the survival times are simulated according to the same distribution in each group, where we consider the standard exponential distribution $\text{Exp}(1)$, a Weibull distribution $\text{Weibull}(1.5,5)$ with parameters $(\lambda_{\text{shape}},\lambda_{\text{scale}})=(1.5,5)$ and a standard log-normal distribution $\text{LogN}(0,1)$. To obtain relevant alternatives, we disturb these null settings by choosing a crossing curve alternative in case of the exponential and lognormal distribution, see Figure~\ref{fig:cross}, and a proportional hazard alternative $\text{Weibull}(1.5,5 (2.5)^{-2/3})$ with hazard ratio equal to $2.5$ for the Weibull distribution. The observations of the first subgroup, $(j_B, j_C)=(1,1)$, for testing of no interaction effect and of the first two subgroups $(j_B, j_C)=(1,1), (1,2)$ for testing of no main effect $B$ are generated according to these alternative distributions while the remaining observations follow the respective null distribution. The censoring times are simulated by uniform distributions $\text{Unif}[0,U_j]$. The upper limit $U_j$ of the interval in group $j$ is determined by a Monte-Carlo simulation such that the average censoring rate $\P( T_{j1} > C_{j1}) = \int_0^{\infty} \min\{x/U_j,1\} \,\mathrm{ d }F_j(x)$ equals a pre-chosen rate $\text{cens}_j$. To cover low, medium and high censoring rates, we consider three different scenarios: $\mbf{\text{cens}}=(7\%, 6\%, 0\%,6\%, 7\%, 0\% )$ (low), $\mbf{\text{cens}} = (20\%,30\%,25\%,35\%,30\%,20\%)$ (med) and $\mbf{\text{cens}} = (20\%,50\%,50\%,20\%,50\%,20\%)$ (high). Thus, the censoring distributions are different implying that the pooled observation pairs ${(X_{ji},\delta_{ji})}_{j,i}$ are non exchangeable. In fact, simulating exchangeable situations would be of less interest as the permutation tests would already be exact testing procedures. In addition, we discuss balanced sample size settings $\mbf{n_1} =(n_{11},\dots,n_{23}) = (8,\ldots,8)$ and $2\cdot \mbf{n_1}$ as well as unbalanced scenarios $\mbf{n_2} =(n_{11},\dots,n_{23})= (15,9,5,9,7,6)$ and $2\cdot\mbf{n_2}$, respectively. The simulations are conducted by means of the computing environment R \citep{R}, version 3.6.2. For each setting, $N_{\text{sim}}= 5000$ simulation runs and $N_{\text{perm}}=1999$ permutation iterations were generated. 

\begin{table}
	\centering
	\caption{Type-1 error rates in $\%$ (nominal level $\alpha = 5\%$) for testing of no main and no interaction effect in the $2\times 3$-layout, respectively, using the joint Wald-type test based on the $\chi^2_f$-approximation (Asy) and the permutation approach (Per). 
		Values inside the $95\%$ binomial interval $[4.4,5.6]$ are printed bold}\label{tab:null_inter}
	\begin{tabular}{lllllllll}
		\multicolumn{3}{c}{}&\multicolumn{2}{c}{low cens.} & \multicolumn{2}{c}{med. cens.} & \multicolumn{2}{c}{high cens.}  \\
		Effect & Distr & $\mbf{n}$ & {Asy} & {Per} & {Asy} & {Per} & {Asy} & {Per}\\
		\hline
		\hline		
		Main & Weib & $\mbf{n}^{(1)}$ & 3.3 & 4.3 & 3.5 & $\mbf{5.4}$& 2.7 & 4.3 \\
		&& $2\mbf{n}^{(1)}$ & $\mbf{5.0}$ & $\mbf{5.2}$ & 4.1 & $\mbf{4.9}$ & 4.1 & $\mbf{4.9}$ \\
		&& $\mbf{n_2}$ & 3.4 & $\mbf{4.9}$ & 2.8 & $\mbf{4.9}$ & 2.4 & $\mbf{4.5}$ \\ 
		&& $2\mbf{n_2}$ & $\mbf{4.6}$ & $\mbf{5.1}$ & 4.3 & $\mbf{5.2}$ & $\mbf{4.5}$ & $\mbf{5.5}$ \\
		%\hline
		& Exp & $\mbf{n}^{(1)}$ & 3.4 & $\mbf{4.7}$ & 3.9 & $\mbf{5.1}$ & 3.2 & $\mbf{4.6}$\\
		&& $2\mbf{n}^{(1)}$ & $\mbf{4.4}$ & $\mbf{4.8}$ & $\mbf{4.4}$ & $\mbf{5.2}$ & $\mbf{4.9}$ & $\mbf{5.3}$  \\
		&& $\mbf{n_2}$ & 3.3 & $\mbf{4.9}$ & 3.0 & $\mbf{4.8}$ & 2.6 & $\mbf{4.7}$ \\ 
		&& $2\mbf{n_2}$ & 4.3 & $\mbf{5.0}$ & 4.1 & $\mbf{4.9}$ & 3.5 & $\mbf{4.7}$ \\
		%\hline
		& logN & $\mbf{n}^{(1)}$ & 3.5 & $\mbf{4.5}$ & 3.7 & $\mbf{5.1}$ & 3.3 & $\mbf{5.3}$\\
		&& $2\mbf{n}^{(1)}$ & $\mbf{4.7}$ & $\mbf{4.6}$ & 4.3 & $\mbf{4.7}$ & 4.2 & $\mbf{4.8}$\\
		&& $\mbf{n_2}$ & 3.0 & $\mbf{4.7}$ & 3.0 & $\mbf{4.5}$ & 2.8 & $\mbf{4.7}$\\ 
		&& $2\mbf{n_2}$ &  4.3 & $\mbf{5.1}$ & 4.1 & $\mbf{4.9}$ & 3.9 & $\mbf{4.7}$\\
		%\hline\hline
		Interaction & Weib & $\mbf{n}^{(1)}$ & 2.0 & $\mbf{4.9}$ & 1.6 & $\mbf{5.0}$ & 1.1 & $\mbf{4.9}$ \\
		&& $2\mbf{n}^{(1)}$ &  3.4 & $\mbf{4.7}$ & 3.0 & $\mbf{4.6}$ & 2.3 & 4.3\\
		&& $\mbf{n_2}$ & 1.6 & 4.3 & 1.8 & $\mbf{5.0}$ & 1.5 & $\mbf{5.3}$ \\ 
		&& $2\mbf{n_2}$ & 3.4 & $\mbf{4.8}$ & 3.0 & $\mbf{5.0}$ & 2.5 & $\mbf{4.5}$ \\
		%\hline
		& Exp & $\mbf{n}^{(1)}$ & 1.9 & $\mbf{5.3}$ & 1.7 & $\mbf{5.0}$ & 1.6 & $\mbf{5.5}$\\
		&& $2\mbf{n}^{(1)}$ & 3.3 & $\mbf{5.1}$ & 3.3 & $\mbf{5.0}$ & 3.0 & $\mbf{5.1}$ \\
		&& $\mbf{n_2}$ & 1.7 & $\mbf{4.4}$ & 1.8 & $\mbf{5.0}$ & 1.7 & $\mbf{5.5}$ \\ 
		&& $2\mbf{n_2}$ & 3.2 & $\mbf{4.9}$ & 2.9 & $\mbf{4.8}$ & 3.3 & $\mbf{5.1}$ \\
		%\hline
		& logN & $\mbf{n}^{(1)}$ & 1.8 & $\mbf{4.8}$ & 1.5 & $\mbf{4.8}$ & 1.2 & $\mbf{4.7}$\\
		&& $2\mbf{n}^{(1)}$ & 3.7 & $\mbf{5.0}$ & 3.4 & $\mbf{5.2}$ & 2.8 & $\mbf{5.2}$\\
		&& $\mbf{n_2}$ & 1.9 & $\mbf{4.8}$ & 1.5 & $\mbf{4.8}$ & 1.8 & $\mbf{5.2}$\\ 
		&& $2\mbf{n_2}$ &  3.3 & $\mbf{5.1}$ & 3.4 & $\mbf{5.0}$ & 2.9 & $\mbf{5.1}$ 	\\
		%\hline 
	\end{tabular}
\end{table}
% 20.02.07_alt for interaction and 20.03.07_alt_oKonp for main effect
\begin{table}[ht] \caption{Power values in $\%$ (nominal level $\alpha = 5\%$) for testing of no main and no interaction effect under crossing (Cross) and proportional (Prop) hazard alternatives, respectively, using the joint Wald-type test based on the $\chi^2_f$-approximation (Asy), the permutation approach (Perm) and the singly-weighted permutation tests based on $w_1$ (LR) and $w_2$ (Cross), respectively} \label{tab:power}
	\centering
	\begin{tabular}{llllcccccccc}
		\multicolumn{4}{c}{}&\multicolumn{4}{c}{$2n^{(1)}$ (bal.)} & \multicolumn{4}{c}{$2n^{(2)}$ (unbal.)}   \\
		Effect& Distr & Altern & cens & Asy  & Per & LR & Cross & Asy & Per & LR & Cross \\
		\hline
		%\multicolumn{2}{c}{}&\multicolumn{8}{c}{one-way layout} \\
		\hline
		Main & Exp & Cross & low & 54.2 & 55.8 & 4.7 & 36.9 & 55.6 & 58.1 & 4.8 & 36.0 \\ 
		&&& med & 47.6 & 49.3 & 4.3 & 31.6 & 48.9 & 51.9 & 5.5 & 31.6 \\ 
		&&& high & 41.0 & 43.5 & 7.4 & 28.6 & 43.5 & 46.9 & 8.6 & 31.2 \\ 
		%\hline
		& LogN & Cross & low & 55.3 & 56.4 & 15.8 & 64.7 & 46.9 & 49.0 & 15.7 & 55.9 \\
		&    && med & 48.9 & 51.0 & 25.6 & 59.2 & 41.6 & 44.5 & 26.1 & 50.1 \\ 
		&   && high & 44.6 & 48.0 & 37.5 & 55.0 & 37.2 & 40.3 & 36.8 & 45.9 \\ 
		%\hline
		& Weib & Prop & low & 65.1 & 66.6 & 77.1 & 42.2 & 70.9 & 73.2 & 81.4 & 30.0 \\ 
		&    &   & med & 49.9 & 52.5 & 64.4 & 42.0 & 56.6 & 60.4 & 70.0 & 37.5 \\ 
		&    &  & high & 36.0 & 39.4 & 51.4 & 39.2 & 43.3 & 48.0 & 58.4 & 41.8 \\
		%\hline
		%\multicolumn{3}{c}{}&\multicolumn{8}{c}{Interaction effect} \\
		%\hline
		Interaction & Exp & Cross &  low & 17.9 & 23.8 & 4.9 & 15.5 & 24.1 & 30.6 & 4.4 & 19.1 \\ 
		&&&  med & 13.6 & 19.1 & 4.6 & 12.5 & 20.0 & 26.3 & 4.6 & 15.8 \\
		&&&  high & 12.1 & 18.1 & 5.7 & 11.9 & 19.1 & 24.8 & 6.6 & 16.8 \\  
		%\hline
		& LogN & Cross &  low & 18.2 & 23.8 & 9.2 & 28.9 & 19.2 & 25.2 & 10.6 & 29.7 \\
		&& &  med & 15.5 & 21.7 & 13.3 & 26.1 & 15.2 & 21.5 & 14.1 & 24.7 \\
		&&&  high & 12.7 & 18.3 & 16.9 & 23.1 & 14.4 & 20.2 & 20.0 & 23.9 \\
		%\hline  
		& Weib & Prop &  low & 19.5 & 26.0 & 38.8 & 18.4 & 37.0 & 44.2 & 55.3 & 16.3 \\  
		&&       & med & 14.7 & 20.7 & 30.7 & 19.0 & 28.4 & 36.1 & 44.1 & 20.2 \\ 
		&&      & high & 9.5 & 15.1 & 23.4 & 18.5 & 20.8 & 27.2 & 35.1 & 27.0 \\ 
		%\hline
	\end{tabular}
\end{table}

Table \ref{tab:null_inter} displays the resulting type-1 error rates. It is apparent that the asymptotic tests lead to conservative decisions with values around $3\%$ for testing of no main effect and even around $2\%$ for testing of no interaction effect for both small sample size settings $\mbf{n_1}$ and $\mbf{n_2}$. The type-1 error rates improve when the sample sizes are doubled, but still stay in a rather conservative range. At the same time, the permutation tests exhibit a satisfactory type-1 error control over all different settings. Only in 4 out of 72 scenarios their type-1 error rates are outside the $95\%$ binomial confidence interval $[4.4\%,5.6\%]$. The results for the one-way layout in the supplement show a different picture: while the permutation tests still control the type-1 error rate accurately, the asymptotic test become very liberal with observed type-1 error rates up to $25.9\%$. 

In Table~\ref{tab:power}, we compare the joint Wald-type statistic $S_n$, its permutation counterpart $S_n^\pi$ and the singly-weighted permutation tests $S_n^\pi(w_1)$ (log-rank weight), $S_n^\pi(w_2)$ (crossing weight) in terms of power. From the simulation results we can draw the following two main conclusions: (1) The conservative type-1 error performance of the asymptotic tests negatively affects their power behavior. Here, the differences between the asymptotic and permutation tests' power values are most pronounced for testing of no interaction effect in the unbalanced settings. (2) The joint Wald-type statistic has a reasonable power behavior for both, proportional and crossing curve alternatives, while choosing the wrong singly-weighted test may lead to a significant power loss. For the exponential distribution settings, the joint Wald-type test even outperforms both singly-weighted tests. This observation can be explained by interpreting the Wald-type statistic $S_n$ as a certain projection as explained in \cite{BrendelETAL2014} for the two sample case. The take-home message from this fact is that combining the weights $w_1$ and $w_2$ as well as combining $w_1$ and $w_3=w_2+\lambda w_1$ for some $\lambda\in R$ results in the same statistic $S_n$. While $w_2$ is designed for crossings near to the center, the hazard rates in the exponential setting actually cross at a mid-early time and, thus, another crossing weight, e.g. $w_3 = w_2-0.25$, would be more appropriate. But, as said before, for the combination approach it does not matter whether we choose $w_2$ or $w_3$, the final result is the same. To sum up, the advantages of the combination approach are that we neither need to choose in advance between proportional and crossing hazard alternatives nor between different crossing points. 

Consequently, we recommend to combine different weights, especially the classical log-rank weight and a crossing weight, rather than blindly guessing in advance which kind of alternative is underlying. Moreover, we prefer the permutation test over the asymptotic test for small sample size scenarios due to the unstable type-1 error rate behavior of the latter with too conservative decisions in the $2\times 3$-layout and too liberal decisions in the one-way layout.

\section{Illustrative data analysis}\label{sec:real_data}
For illustration purposes, we re-analyzed the lung cancer study from \cite{Prentice}, which is freely available in the R package survival. It includes the survival times of male lung cancer patients getting either an experimental or a standard treatment. As statistically verified by \cite{KalbfleischPrentice1980}, the tumor type has an effect on the patients' survival. Thus, \cite{ditzhaus:pauly:2019} restricted to the smallcell tumor type to illustrate their two-sample combination approach. The general factorial design set-up allows us to inspect, now, several tumor types simultaneously. In detail, we consider a $2\times 3$-layout with the factors treatment, having the two levels experimental and standard, and tumor type with the three different levels: smallcell, adeno and large. For the experimental treatment, the group-specific sample size $n_j$ and the censoring rates $\text{cens}_j$ are $(n_j,\text{cens}_j) = (18,6\%), (18,6\%),(12,0\%)$ and, for the standard treatment, we have $(n_j,\text{cens}_j)=(30,7\%),(9,0\%),(15,6\%)$, respectively, where the first values correspond to the tumor type smallcell, the second to adeno and the last to large. This situation is comparable with the unbalanced sample size scenario $2\mbf{n_2}$ combined with the low censoring setting. In Figure~\ref{fig:data_exam}, the survival curves of the two treatments are displayed for all the three tumor types. It appears that the experimental treatment has a beneficial effect on the patients' survival time. This was already statistically verified by \cite{ditzhaus:pauly:2019} for the smallcell tumor type and can be confirmed when considering all three tumor types simultaneously by the joint Wald-type test and the singly-weighted test based on $w_1$, see Table~\ref{tab:data_examp_WTS}. For comparison, an Aalen- and a Cox-regression including the different tumor types as dummy variables result in a slightly non-significant $p$-value of $7.8\%$ and a significant $p$-value of $1.9\%$ for the treatment effect, respectively. Contrary to all these results, the singly-weighted test based on $w_2$ has a very high $p$-value not supporting a rejection. These diverse decisions reaffirm our recommendation from the simulation section, namely to combine the weights rather than blindly choosing one weight in advance. 

\begin{figure}
	%graphics were generated in the folder 2019.10.08_median_power
	%\figurebox{10cm}{5cm}{}[small.eps]
	\begin{minipage}{0.32\textwidth}
		\includegraphics[width=\textwidth]{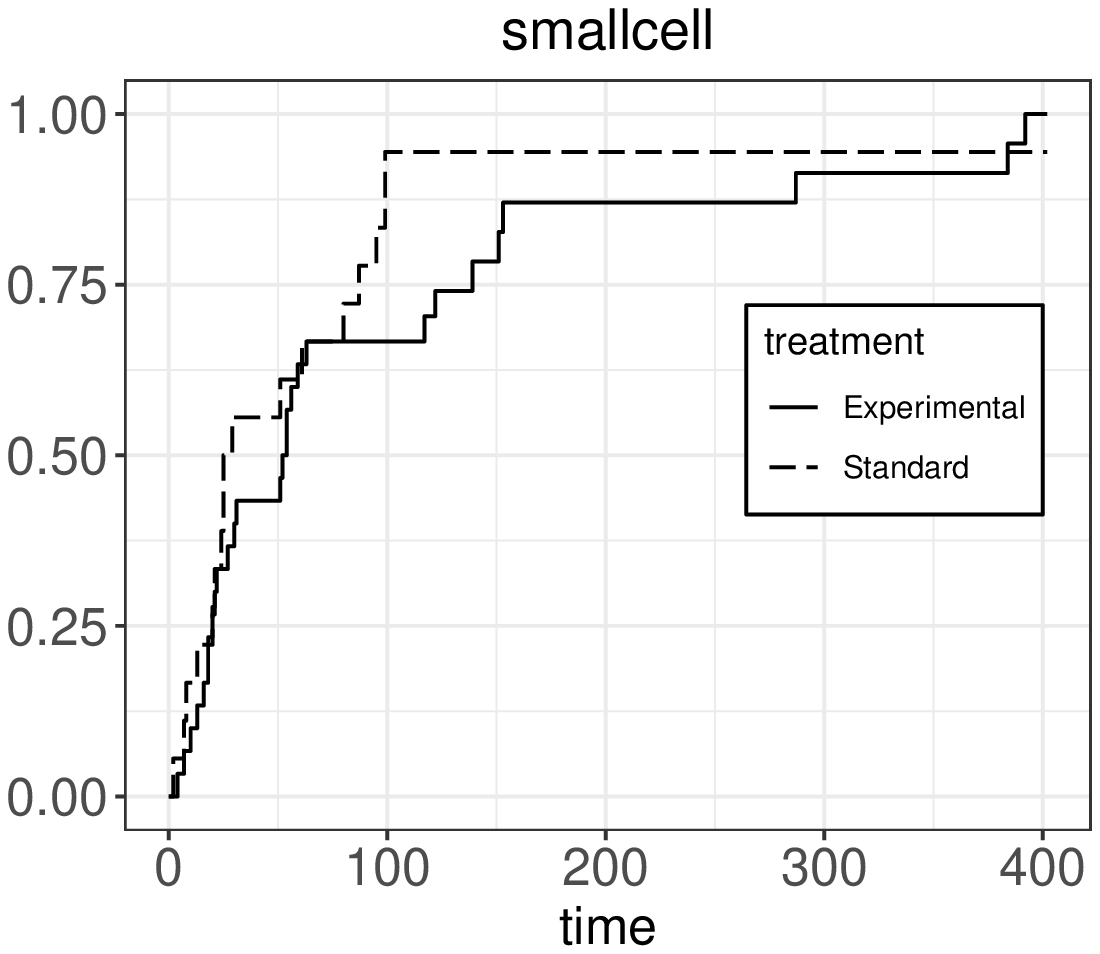}
	\end{minipage}
	\begin{minipage}{0.32\textwidth}
		\includegraphics[width=\textwidth]{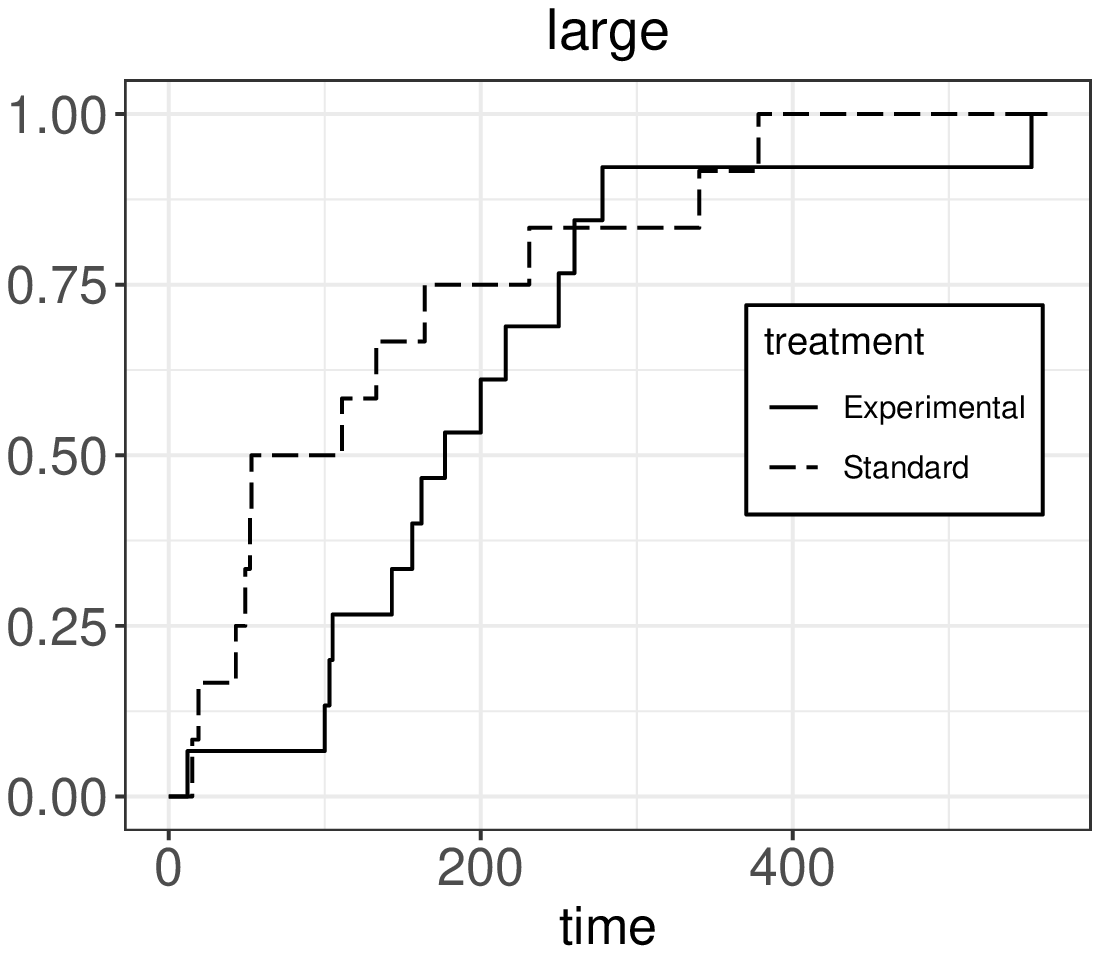}
	\end{minipage}
	\begin{minipage}{0.32\textwidth}
		\includegraphics[width=\textwidth]{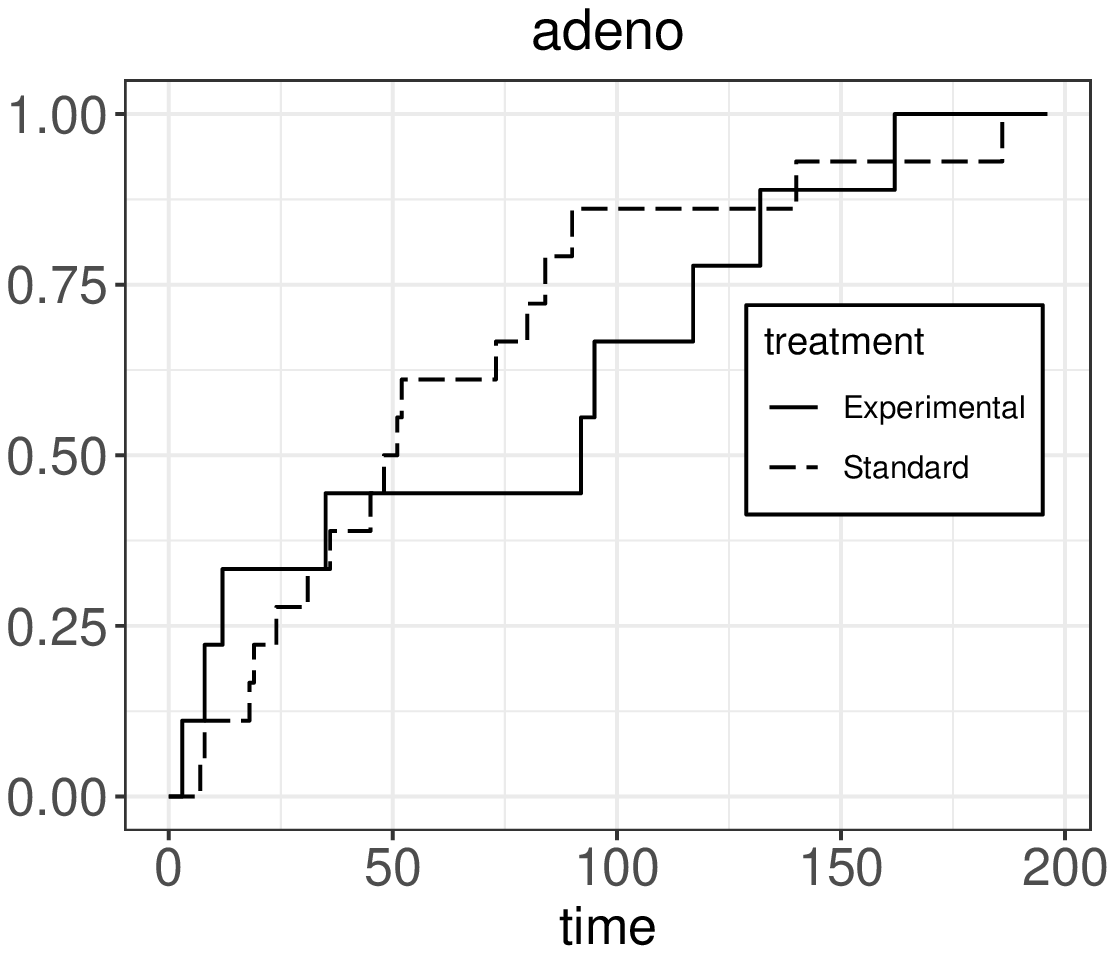}
	\end{minipage}
	\caption{Kaplan--Meier curves for the lung-cancer trial}\label{fig:data_exam}
\end{figure}

%\begin{table}[ht] \caption{Regression mit smallcell als Referenzgruppe}\label{tab:aalen+cox}
%	\centering
%	\begin{tabular}{l|ccc|ccc}
%		\multicolumn{1}{c}{} & \multicolumn{3}{c}{Aalen regr.} & \multicolumn{3}{c}{Cox regr.} \\
%        Covariate & Estimate & $p$-value & $95\%$ CI & Estimate & $p$-value & $95\%$ CI  \\
%		 Treatment & $7.1\cdot 10^{-3}$ & 7.8 & [$6.3*10^{-3}$, $7.9\cdot 10^{-3}$] &   0.52 & 1.9  & [0.48, 0.56] %		 large & -0.013 & 0.1 &	$[-0.014, -0.012]$ & -0.92 & $<0.1$ & [-0.97, -0.87]\\
%		 adeno & -0.000389 & 94.4 &	$[ -0.0015,  0.0007]$	&	-0.01 & 96.3 & [-0.06,  0.04]
%	\end{tabular}
%\end{table}

\begin{table}[ht] \caption{$P$-values in $\%$ for the lung cancer data example of the joint Wald-type test (Comb) as well as the singly-weighted tests based on $w_1$ (LR) and $w_2$ (Cross)} \label{tab:data_examp_WTS}
	\centering
	\begin{tabular}{lcccccc}
		\multicolumn{1}{c}{} & \multicolumn{3}{c}{Asymptotic} & \multicolumn{3}{c}{Permutation} \\
		& Comb & LR & Cross & Comb & LR & Cross  \\
		\hline \hline
		Treatment & 1.3 & 2.6 & 79.1 & 1.0 & 2.8  & 79.1 \\
		celltype & 0.2 & $<0.1$ & 1.5 & 0.02 & $<0.1$ & 1.4\\
		interaction & 72.7 & 99.0 &	68.4 & 75.0 & 99.2 & 69.1\\
		\hline
	\end{tabular}
\end{table}

%\begin{figure}[ht] 
%	\begin{center}
%		\begin{minipage}{0.45\textwidth}
%			\includegraphics[width=\textwidth]{Grafics/Exponent.eps}
%		\end{minipage}
%		\begin{minipage}{0.45\textwidth}
%			\includegraphics[width=\textwidth]{Grafics/Weib.eps}
%		\end{minipage}
%		\vspace{-0.05cm}	
%		
%		\begin{minipage}{0.45\textwidth}
%			\includegraphics[width=\textwidth]{Grafics/Logn.eps}
%		\end{minipage}
%	\end{center}	
%\end{figure}

\section{Outlook}

The proofs for the asymptotic test versions, i.e., Theorems \ref{theo:null_unc} \& \ref{theo:several_weights}, rely on the fact that $N=(N_1,\ldots,N_k)$ fulfills the multiplicative intensity model of \cite{aalen:1978}. More complex filtering mechanisms, e.g., truncation or certain interval censorings, can be endowed in the same methodology \citep[Chapter III]{ABGK} and an extension of the proofs of Theorems~\ref{theo:null_unc}-\ref{theo:several_weights} is straightforward. For such more general survival settings, however, it is unclear whether the permutation technique is still valid. 
That is why multiplier resampling as investigated in \cite{Lin1997,BeyersmannETAL2013,
	%dobler2014bootstrapping,
	dobler2017non,dobler2019confidence} and \cite{bluhmki2018wild,bluhmki2019wildNA} would be our first choice to approximate the asymptotic $\chi_f^2$-quantile in the cases beyond (pure) right-censoring. 

To bring the presented procedures into statistical practice, the authors currently work on two fields: (1) an implementation of the combination test into the R-package GFDsurv, which will be available on CRAN soon. The corresponding R-function is coined CASANOVA abbreviating the presented cumulative Aalen survival analysis-of-variance approach, and (2) on convincing 
medical doctors and epidemiologists to apply the methods in bio-statistical co-operations.

\appendix

\section{Additional simulations}
We run additional simulations for a one-layout comparing $k=6$ groups under the settings from Section \ref{sec:simu}. For the power comparison, we follow the strategy for testing no main effects, i.e. the observations from the first two groups are generated by proportional or crossing hazard alternatives while the remaining observations follow the respective null setting. Table~\ref{tab:null_one-way} displays the results under the null hypothesis $\mathcal H_0: \{\mbf{P}_k\mbf{A} = \mbf{0}_k\}$. While the permutation test controls again the type-1 error rate accurately, the decisions of the asymptotic test are now very liberal with values around $12\%$ for the balanced setting $\mbf{n}_1$ and even around $20\%$ for the unbalanced scenario $\mbf{n}_2$. Doubling the sample size improves the type-1 error rates and they come closer to the $5\%$-benchmark but are still in a rather liberal range with values up to $18.5\%$ in the unbalanced case $2\mbf{n}_2$. Consequently, a comparison of the asymptotic and permutation test in terms of power is unfair. Nevertheless, the power values of both tests as well as of the singly-weighted permutation approaches $S_n^\pi(w_1)$ and $S_n^\pi(w_2)$ are displayed in Table~\ref{tab:power_one-way}. Ignoring the blown-up power values of the asymptotic test, the results support the recommendation from Section~\ref{sec:simu} that the joint Wald-type test combines the strength of both singly-weighted tests and that the crossing point do not need to be chosen in advance.

\begin{table}
	\centering
	\caption{Type-1 error rate in $\%$ ($\alpha = 5\%$) for testing of no group effect in the one-way layout $(k=6)$ using the joint Wald-type test based on the $\chi^2_f$-approximation (Asy) and the permutation approach (Per). The superscript $^*$ highlights values inside the $95\%$ binomial interval $[4.4,5.6]$.}\label{tab:null_one-way}
	\begin{tabular}{ll|ll|ll|ll}
		\multicolumn{2}{c}{}&\multicolumn{2}{c}{low cens.} & \multicolumn{2}{c}{med. cens.} & \multicolumn{2}{c}{high cens.}  \\
		Distr & $\mbf{n}$ & {Asy} & {Per} & {Asy} & {Per} & {Asy} & {Per}\\
		\hline	\hline	
		Exp  & $\mbf{n}_1$  & 11.1 & {4.7}$^*$ & 13.5 & ${5.3}^*$ & 12.7 & {5.0}$^*$ \\ 
		& $2\mbf{n}_1$  & 6.6 & {5.0}$^*$  & 8.3 & {5.5}$^*$ & 10.8 & {4.9}$^*$ \\
		& $\mbf{n}_2$  & 17.4 & {4.7}$^*$ & 23.7 & {4.7}$^*$ & 25.8 & {4.9}$^*$ \\ 
		& $2\mbf{n}_2$  & 8.8 & {4.5}$^*$ & 10.9 & {4.8}$^*$ & 17.2 & 6.4 \\
		\hline
		LogN  & $\mbf{n}_1$  & 10.8 & {4.8}$^*$ & 13.2 & {4.8}$^*$ & 11.6 & ${5.4}^*$ \\ 
		& $2\mbf{n}_1$  & 7.1 & ${5.4}^*$ & 7.5 & {4.6}$^*$ & 11.1 & ${5.3}^*$ \\ 
		& $\mbf{n}_2$  & 18.6 & {4.8}$^*$ & 23.0 & {4.5}$^*$ & 25.9 & 5.8 \\
		& $2\mbf{n}_2$  & 8.7 & {4.7}$^*$ & 11.2 & {4.9}$^*$ & 17.4 & 6.4 \\
		\hline
		Weib & $\mbf{n}_1$  & 12.0 & {5.5}$^*$ & 13.4 & {5.2}$^*$ & 9.1 & {4.7}$^*$ \\ 
		& $2\mbf{n}_1$  & 6.7 & {5.1}$^*$ & 8.8 & {5.2}$^*$ & 11.4 & {5.1}$^*$ \\ 
		& $\mbf{n}_2$  & 17.4 & 4.2 & 24.8 & {4.5}$^*$ & 23.5 & {5.2}$^*$ \\  
		& $2\mbf{n}_2$  & 8.8 & {4.7}$^*$ & 12.0 & {5.5}$^*$ & 18.5 & 6.4 \\ 		
	\end{tabular}
\end{table}

% 20.03.07_alt_oKonp
\begin{table}[ht] \caption{Power values in $\%$ ($\alpha = 5\%$) for testing of no group effect in the one-way layout $(k=6)$ under crossing (Cross) and proportional (Prop) hazard alternatives, respectively, using the joint Wald-type test based on the $\chi^2_f$-approximation (Asy), the permutation approach (Per) and the singly-weighted permutation tests based on $w_1$ (LR) and $w_2$ (Cross), respectively} \label{tab:power_one-way}
	\centering
	\begin{tabular}{lll|cccc|cccc}
		\multicolumn{3}{c}{}&\multicolumn{4}{c}{$2\mbf{n}_1$ (bal.)} & \multicolumn{4}{c}{$2\mbf{n}_2$ (unbal.)}   \\
		Distr & Altern & cens & Asy & Per & LR & Cross & Asy & Per & LR & Cross \\
		\hline
		\hline
		Exp & Cross & low & 67.4 & 61.7 & 5.2 & 24.3 & 77.2 & 62.9 & 4.4 & 35.2 \\  
		&& med & 56.4 & 44.8 & 4.5 & 16.8 & 67.8 & 44.5 & 6.6 & 24.6 \\ 
		&& high & 46.3 & 28.6 & 6.4 & 14.0 & 61.6 & 29.6 & 8.5 & 20.5 \\ 
		\hline
		LogN & Cross & low & 81.3 & 76.3 & 19.3 & 86.9 & 76.3 & 59.6 & 20.3 & 74.6 \\ 
		&& med & 81.7 & 72.1 & 31.6 & 80.7 & 77.0 & 43.0 & 33.4 & 57.4 \\ 
		&& high & 83.8 & 61.4 & 48.6 & 72.3 & 74.5 & 19.2 & 47.2 & 31.9 \\
		\hline
		Weib & Cross & low & 47.1 & 40.5 & 64.8 & 31.0 & 75.4 & 59.6 & 83.3 & 29.5 \\ 
		&& med & 40.2 & 29.5 & 47.6 & 25.6 & 65.0 & 40.5 & 66.5 & 29.5 \\  
		&& high & 37.2 & 21.1 & 33.0 & 21.3 & 60.1 & 29.9 & 47.9 & 27.2 \\ 
	\end{tabular}
\end{table}

\section{Preliminaries}
The uniform convergence of the processes $Y_j,N_j$ as well as of the pooled Kaplan--Meier estimator are well known. Since the proofs are short, we present them here for completeness reasons. Define $\tau_j = \inf\{x\geq 0: [1-F_j(x)][1-G_j(t)] = 0\}\in\R \cup\{\infty\}$ $(j=1,\ldots,k)$ and $\tau = \min(\tau_j:j=1,\ldots,k)$.
\begin{lemma}\label{lem:unif_conv_w}
	Define $\nu_j(t) = \kappa_j\int_0^t [1-G_j(s)]\,\mathrm{ d }F_j(s)$ $(j=1,\ldots,k)$, $\nu=\sum_{j=1}^k\nu_j$ and
	\begin{align*}
	F_0(t) = 1 - \exp\Bigl\{ - \sum_{j=1}^k  \int_0^t \frac{1}{y(s)} \,\mathrm{ d }\nu_j(s) \Bigr\}.
	\end{align*}
	Then for every $\tau_0\in(0,\tau)$ we have in probability under the local alternatives \eqref{eqn:local_alternatives} of Section \ref{sec:local}
	\begin{align*}
	\sup_{t\in[0,\tau_0]} |\widehat F_{n}(t) - F_0(t)| + \sup_{t\in[0,\infty)}  | Y_j(t)/n - y_j(t) | + \sup_{t\in[0,\infty)} \Bigl | N_j(t)/n - \nu_j(t)\Bigr|  \to 0.
	\end{align*}	
\end{lemma}
\begin{proof}[of Lemma \ref{lem:unif_conv_w}]
	Let $F_{nj}$ $(j=1,\ldots,k)$ be the distribution function belonging to $A_{nj}$, i.e., $1-F_{nj}(t)=\exp\{-A_{nj}(t)\}$ $(t\geq 0)$. By straight forward calculations $F_{nj}(t)\to F_j(t)$ for all $t\in\R$. Combining this and Chebyshev's inequality we obtain $n^{-1}Y_j(t)\to y_j(t)$ and $N_j(t)\to \nu_j(t)$, both in probability. Since all functions are non-decreasing and the limits are continuous it is well known that both convergences hold even uniformly, which proves the convergence statements corresponding to $Y_j/n$ and $N_j/n$. From this and 
	\begin{align}\label{eqn:KME_cons}
	-\log \{ 1 - \widehat F_n(t) \} = \int_0^t \log \Bigl\{ \Bigl( 1- n^{-1}\frac{n}{Y(s)} \Bigr)^n \Bigr\} \,\mathrm{ d }\frac{N}{n}(s)
	\end{align}
	we obtain, finally, uniform convergence of the pooled Kaplan--Meier estimator to $F_0$. In the two-sample situation, \cite{neuhaus:1993}, see p.1773,  already used the representation in \eqref{eqn:KME_cons} to prove convergence of $\widehat F_n$ to $F_0$.
\end{proof}

\section{Proofs of Section \ref{sec:asy_properties}}

To avoid unnecessary repetitions, we simultaneously prove Theorems \ref{theo:null_unc}--\ref{theo:several_weights}. For this purpose, we directly consider local alternatives \eqref{eqn:local_alternatives} from Section~\ref{sec:local}, where the null hypothesis is covered by setting $\gamma_j\equiv 0$ for all $j=1,\ldots,k$. In the main paper, we restricted to polynomial $\widetilde w_r\in\mathcal W$ but the proofs are valid for general weights $\widetilde w_r\in \mathcal W$ fulfilling the following weaker assumption:
\begin{assumption}\label{ass:lin_ind_suppl}
	There is a number $L$ such that the functions $\widetilde w_1,\ldots,\widetilde w_m$ are linearly independent on any subset of $(0,\tau)$ with at least $L$ different elements.
\end{assumption}
Obviously, Assumption \ref{ass:lin_ind} implies Assumption \ref{ass:lin_ind_suppl}. Recall the definition of the centred Nelson--Aalen-type integrals
\begin{align*}
\widetilde Z_{nj}(w_{nr}) = n^{1/2}\int_0^\infty w_{nr}(t) \,\mathrm{ d }(\widehat A_j - A_j)(t), \quad \mbf{Z}_n(w_{nr}) = \{Z_{n1}(w_{nr}),\ldots,Z_{n1}(w_{nr})\}\trans.
\end{align*}
Subsequently, we prove a central limit theorem for $\mbf{\widetilde Z}_n=\{ \mbf{\widetilde Z}_n(w_{n1})\trans,\ldots,\mbf{\widetilde Z}_n(w_{nm})\trans\}\trans$ as well as the consistency of the covariance matrix estimator $\mbf{\widehat \Sigma}$ introduced in Section \ref{sec:comb_weights}. 
\begin{theorem}\label{theo:null_unc_general}
	The statistic $\mbf{\widetilde Z}_n$ converges in distribution to a multivariate normal distribution $N(\mbf{\mu}, \mbf{\Sigma})$ with expectation (row) vector $\mbf{\mu}=(\mu_{11},\ldots,\mu_{1k},\mu_{21},\ldots,\mu_{mk})\trans\in\R^{mk}$ defined by
	\begin{align*}
	\mu_{rj} = \int_0^\infty  w_r(t) \gamma_j(t) \,\mathrm{ d }A_j(t),\quad w_r(t) = \widetilde w_r\{F_0(t)\}\frac{y_1(t)\ldots y_k(t)}{y(t)^{k-1}},
	\end{align*}
	and with regular covariance (block-)matrix $\mbf{\Sigma}= (\mbf{\Sigma}^{(rr')})_{r,r'=1,\ldots,m}\in\R^{km\times km}$, where each submatrix $\mbf{\Sigma}^{(rr')} = \text{diag}(\sigma^{2,(rr')}_1,\ldots,\sigma^{2,(rr')}_k)$ is a diagonal matrix with entries
	\begin{align}\label{eqn:def_Sigma}
	\sigma^{2,(rr')}_j = \int_0^\infty  \frac{w_r(t) w_{r'}(t)}{ y_j(t)} \,\mathrm{ d }A_j(t) \quad (j=1,\ldots,k).
	\end{align}
\end{theorem}
\begin{lemma}\label{lem:sig_cons_suppl}
	The estimator $\mbf{\widehat \Sigma}$ is consistent for $\mbf\Sigma$, i.e., $\widehat\sigma^{2,(rr')}_j \to \sigma^{2,(rr')}_j$ in probability.
\end{lemma}
All results mentioned in Section \ref{sec:asy_properties} follow from Theorem \ref{theo:null_unc_general} and Lemma \ref{lem:sig_cons_suppl}. In particular, the convergence statements of the Wald-type statistics $S_n(w_n)$ and $S_n$ in Theorems \ref{theo:local_alternatives} \& \ref{theo:several_weights}, respectively, follow from Theorem \ref{theo:null_unc_general}, Lemma \ref{lem:sig_cons_suppl}, the continuous mapping theorem and Theorem 9.2.2 of \cite{rao:mitra:1971}.

\subsection{Proof of Theorem \ref{theo:null_unc_general} and Lemma \ref{lem:sig_cons_suppl}}
For the proof, we combine martingale theory and the counting process approach. We refer the reader to \cite{ABGK} for a detailed introduction to both. For our purposes, their Chapters II.5.1 and III.3 are mainly relevant. For a certain filtration, which we do not want to specific here, the processes $t\mapsto w_{nr}(t)$ are predictable and $(N_1,\ldots,N_k)\trans$ fulfills the multiplicity intensity model \citep{ABGK} with intensity process $(\lambda_{n1},\ldots,\lambda_{nk})\trans$,
\begin{align*}
\lambda_{nj} = Y_j \alpha_{nj}\quad (j=1,\ldots,k).
\end{align*}
In particular, $\widetilde M_{nj}(t) = N_j(t) - \int_0^{t} Y_j \,\mathrm{ d }A_{nj}$ and, thus,
\begin{align}\label{eqn:martingale}
t \mapsto [M_{nrj}(t)]_{r=1,\ldots,m;j=1,\ldots,k} = \Bigl[ \int_0^t n^{1/2} w_{nr}(s) \frac{1}{Y_j(s)}\,\mathrm{ d }\widetilde M_{nj}(s) \Bigr]_{r=1,\ldots,m;j=1,\ldots,k}
\end{align}
are (multivariate) local square integrable martingales. Let $K_{r}=\sup\{|\widetilde w_r(t)|:t\in\R\}$ $(r=1,\ldots,m)$ and $K=\max\{K_{r}:r=1,\ldots,m\}$. Then $|w_{nr}/(Y_j/n)|$ is uniformly bounded by $K$ and, thus, the integrand in \eqref{eqn:martingale} is bounded by $n^{-1/2}K$. Consequently, the Lindeberg condition is always fulfilled for $M_{nrj}$. By Rebolledo's Theorem it remains to discuss the predictable covariation processes (only at the end point $t=\infty$)
\begin{align}\label{eqn:quadra_variat}
&\langle M_{nrj} , M_{nr'j} \rangle = n\int_0^\infty w_{nr}(t)w_{nr'}(t) \frac{1}{Y_j^2(t)} \alpha_{nj}(t) \,\mathrm{ d }t \nonumber \\
&= \int_0^\infty \frac{w_{nr}(t)w_{nr'}(t)}{n^{-2}Y_j(t)^2} \frac{Y_j(t)}{n}\alpha_j(t) \,\mathrm{ d }t + n^{-1/2}  \int_0^\infty \frac{w_{nr}(t)w_{nr'}(t)}{n^{-2}Y_j(t)^2} \frac{Y_j(t)}{n}\gamma_j(t)\alpha_j(t) \,\mathrm{ d }t.
\end{align}
Due to the underlying independence between the groups we have $\langle M_{nrj} , M_{nr'j'} \rangle = 0$ for $j\neq j'$. In the following, we will show 
\begin{align}\label{eqn:int_gill}
\int_0^\infty \frac{w_{nr}(t)w_{nr'}(t)}{n^{-2}Y_j(t)^2} \frac{Y_j(t)}{n} \alpha_j(t) \,\mathrm{ d }t \to \int_0^\infty w_{r}(t)w_{r'}(t) \frac{1}{y_j(t)} \alpha_j(t) \,\mathrm{ d }t
\end{align}
in probability. Denote by $I_{nrr'j}$ and $I_{rr'j}$ the integrands on the left and right hand side, respectively. Note that the weighting functions cause $I_{nrr'j}(t)=0=I_{rr'j}(t)$ for all $t\in[\tau,\infty)$, when $\tau<\infty$. By Lemma \ref{lem:unif_conv_w} $I_{nrr'j}$ converges pointwisely to $I_{rr'j}$. Thus, by a result of Gill \citep[Prop. II.5.3]{ABGK}, it is sufficient for \eqref{eqn:int_gill} to show that $\int_0^\infty |I_{rr'j}(t)| \,\mathrm{ d }t<\infty$ and that there is an integrable function $g_\lambda$ for every $\lambda>1$ such that
\begin{align}\label{eqn:suff_gill}
\P \Bigl( |I_{nrr'j}(t)| \leq g_\lambda(t) \text{ for all }t\Bigr) \geq 1 - \frac{e}{\lambda}.
\end{align}
Since $\alpha_j =f_j /(1-F_j)$ we obtain
\begin{align}\label{eqn:calcu1}
\int_0^\infty  | I_{rr'j}(t) | \,\mathrm{ d }t \leq K^2 \int_0^\infty y_j(t) \alpha_j(t) \,\mathrm{ d }t = K^2 \int_0^\infty \{1-G_j(t)\}\,\mathrm{ d }F_j(t) < \infty.
\end{align}
Moreover, we obtain from Remark 1(i) of \cite{wellner:1978} that for all $\lambda\geq 1$
\begin{align*}
\P \Bigl( \sup_{t\in[0,\tau)}\Bigl|  \frac{n_j^{-1}Y_j(t)}{\{1-G_j(t)\}\{1-F_j(t)\}} \Bigr| \leq\lambda  \Bigr) \geq 1 - \frac{e}{\lambda}.
\end{align*}
Since $n_j/n\leq 2\kappa_j$ for sufficiently large $n$, we can deduce that \eqref{eqn:suff_gill} is fulfilled for the integrable function $g_\lambda= {1}_{[0,\tau)} 2K^2\lambda y_j \alpha_j$, the integrability follows from the calculation in \eqref{eqn:calcu1}. 
Analogously, we can deduce that the second summand in \eqref{eqn:quadra_variat} converges to $0$. Finally, 
\begin{align*}
\langle M_{nrj} , M_{nr'j} \rangle \to \int_0^\infty  \frac{w_{r}(t)w_{r'}(t)}{y_j(t)} \alpha_j(t) \,\mathrm{ d }t = \sigma^{2,(rr')}_j.
\end{align*}
in probability. The aforementioned Theorem of Rebolledo implies distributional convergence of $(M_{nrj}(\infty))_{r=1,\ldots,m;j=1,\ldots,k}$ to a centred multivariate normal distribution with covariance matrix $\mbf{\Sigma}$ and also convergence of the optional covariation process to $\mbf{\Sigma}$, i.e.,
\begin{align*}
n\int_0^\infty w_{nr}(t)w_{nr'}(t) \frac{1}{Y_j^2(t)} \,\mathrm{ d }N_j(t) = [M_{nrj} , M_{nr'j}] \to \sigma^{2,(rr')}_j
\end{align*}
in probability. The latter implies Lemma \ref{lem:sig_cons_suppl}. \\
In the general situation of local alternatives, we have 
\begin{align*}
\widetilde Z_{nj}(w_{nr}) = M_{nrj}(\infty) + \int_0^\infty w_{nr}(t) \gamma_j(t) \alpha_j(t) \,\mathrm{ d }t.
\end{align*}
Using the same strategy as for \eqref{eqn:int_gill}, we can prove that the second integrand converges in probability to $\mu_{rj}$. Consequently, the convergence of $\widetilde Z_n$ follows from Slutzky's Lemma.

It remains to prove the regularity of $\mbf{\Sigma}$. Therefor, let $\mbf{\beta}=(\beta_{11},\ldots,\beta_{1k},\beta_{21},\ldots,\beta_{km})\trans\in \R^{km}$ such that $\mbf{\beta}\trans\mbf{ \Sigma}\mbf{\beta} = 0$. Then
\begin{align}\label{eqn:betaSigmabeta=0}
0=\mbf{\beta}\trans\mbf{ \Sigma}\mbf{\beta} &= \sum_{j=1}^k \sum_{r=1}^m\sum_{r'=1}^m \beta_{rj} \beta_{r'j} \int_0^\infty \frac{w_r(t)w_{r'}(t)}{y_j(t)} \,\mathrm{ d }A_j(t) \nonumber \\
& = \sum_{j=1}^k \kappa_j  \int_0^\tau \Bigl[ \sum_{r=1}^m \beta_{rj} \widetilde w_r\{F_0(t)\} \Bigr]^2\frac{y_1(t)\ldots y_k(t)}{y(t)y_j(t)^2} \,\{1-G_j(t)\}\mathrm{ d }F_j(t).
\end{align}
Since $F_j(t)>F_j(s)$ always implies $F_0(t)>F_0(s)$, we can deduce for every $j=1,\ldots,k$ from $F_j(\tau)>0$ and \eqref{eqn:betaSigmabeta=0} that $\sum_{r=1}^m \beta_{rj} \widetilde w_r\{F_0(t)\}=0$ for infinitely many different $t$. Hence, $\beta_{rj}=0$ follows from Assumption \ref{ass:lin_ind_suppl}, which proves the regularity of $\mbf\Sigma  $.

\section{Proof of Theorem \ref{theo:perm}}
Let $X_{(1)}<\ldots<X_{(n)}$ be the order statistics of the pooled observations. Denote by $\delta_{(i)}$ the corresponding censoring status and $c_{(i)}$ the group membership, i.e., $c_{(i)}=j$ iff $X_{(i)}$ belongs to group $j$. By Lemma \ref{lem:unif_conv_w} we have in probability
\begin{align}\label{eqn:unif_Y+N}
\sup_{t\in[0,\infty)} \Bigl | \frac{Y(t)}{n} - y(t) \Bigr| + \sup_{t\in[0,\infty)} \Bigl | \frac{N(t)}{n} - \nu(t)\Bigr| \to 0.
\end{align}
By restricting to appropriate subsequence, we can suppose that \eqref{eqn:unif_Y+N} holds with probability one, even when we consider local alternatives. Without loss of generality, we work along such subsequences for the remaining proof. From now on, let the observations be fixed. The permutation approach affects just the group membership $c_{(i)}$ of $X_{(i)}$ and not the censoring status $\delta_{(i)}$. That is why we can treat $N$, $Y$ and $\widehat F_n$ as fixed functions. We can assume without loss of generality that \eqref{eqn:unif_Y+N} holds and, thus, by \eqref{eqn:KME_cons} we have $\sup_{t\in[0,\tau_0]}|\widehat F_n(t)- F_0(t)| \to 0$ for all $\tau_0$ with $y(\tau_0)>0$, where $F_0$ is specified in Lemma \ref{lem:unif_conv_w}.  We add the superscript $^\pi$ to the quantities actually depending on the permuted data, i.e., we write $\mbf{Z}_n^\pi,Y_j^\pi, c_{(i)}^\pi, N_j^\pi$ etc.

For the proof, we write all relevant quantities as sums and use discrete martingale techniques to derive the asymptotic distribution. Similarly to the proof for the asymptotic test, we do not derive a central limit theorem for $\mbf{Z}_n^\pi$ itself but for $\mbf{\bar{Z}}^\pi_n=[\bar Z_{1}(w_{n1}^\pi),\ldots,\bar Z_{k}(w_{n1}^\pi),\bar Z_{1}(w_{n2}^\pi),\ldots,\bar Z_{k}(w_{nm}^\pi)]\trans\in\R^{km}$, a centred version of it, with elements 
\begin{align}\label{eqn:def_tilde_Z_pi}
\bar Z_{j}(w_{nr}^\pi) &= n^{1/2}\int_0^\tau w_{nr}^\pi(t) \,\mathrm{ d }( \widehat A_j - \widehat A)(t) \nonumber \\
&= \sum_{i=1}^n n^{1/2}w_{nr}^\pi(X_{(i)}) \Bigl( \frac{\Delta N_j^\pi(X_{(i)})}{Y_j^\pi(X_{(i)})} - \frac{\Delta N(X_{(i)})}{Y(X_{(i)})}  \Bigr).
\end{align}
It is easy to see that $\mbf{T}^{(m)} \mbf{Z}^\pi_n = \mbf{T}^{(m)} \mbf{\bar{Z}}^\pi_n$. Now, introduce the filtration
\begin{align*}
\mathcal F_{n,i} = \sigma( c_{(1)}^\pi,\ldots,c_{(i)}^\pi).
\end{align*}
Clearly, $w_n^\pi(X_{(i)})$ and $Y_j^\pi(X_{(i)})$ are predictable under this filtration. Moreover, 
\begin{align*}
\E\Bigl( \Delta N_j^\pi(X_{(i)}) \mid \mathcal F_{n,i-1} \Bigr) = \delta_{(i)}\E\Bigl( I\{c_{(i)}^\pi = j\} \mid \mathcal F_{n,i-1} \Bigr) = \delta_{(i)}\frac{Y_j^\pi(X_{(i)})}{Y(X_{(i)}) }
\end{align*}
and, thus, the summands in \eqref{eqn:def_tilde_Z_pi} form indeed a martingale difference scheme. Since the summands in \eqref{eqn:def_tilde_Z_pi} are uniformly bounded by $n^{-1/2}K$, where $K$ is defined as in the previous proof, the Lindeberg condition is always fulfilled. Again, it just remains to discuss the predictable covariation process given by
\begin{align*}
C^{(rr')}_{njj'} =&\,	n\sum_{i=1}^n \E \Bigl[ w_{nr}^\pi(X_{(i)})w_{nr'}^\pi(X_{(i)})\Bigl\{ \frac{\Delta N_j^\pi(X_{(i)})}{Y_j^\pi(X_{(i)})} - \frac{\Delta N(X_{(i)})}{Y(X_{(i)})} \Bigr\} \\
&\phantom{n\sum_{i=1}^n \E \Bigl[}\cdot\Bigl\{ \frac{\Delta N^\pi_{j'}(X_{(i)})}{Y^\pi_{j'}(X_{(i)})} - \frac{\Delta N(X_{(i)})}{Y(X_{(i)})}\Bigr\} \mid \mathcal F_{n,i-1} \Bigr] \\
=& \,n \sum_{i=1}^{n} w_{nr}^\pi(X_{(i)})w_{nr'}^\pi(X_{(i)}) \Delta N(X_{(i)}) \Bigl\{ \frac{I\{j=j'\}}{Y_j^\pi(X_{(i)})Y(X_{(i)})} - \frac{1}{Y(X_{(i)})^2} \Bigr\} \\
=& \int_0^\infty w_{nr}^\pi(t)w_{ns}^\pi(t)\Bigl\{ \frac{I\{j=j'\}}{n^{-2}Y_j^\pi(t)Y(t)} - \frac{1}{n^{-2}Y(t)^2} \Bigr\} \,\mathrm{ d }\frac{N}{n}(t).
\end{align*}
Fix $\tau_0\in(0,\infty)$ such that $y(\tau_0)>0$. Then \cite{neuhaus:1993} showed, see his Equation (6.1), that 
\begin{align*}
\sup_{t\in[0,\tau_0]} \Bigl | \frac{Y_j^\pi(t)}{Y(t)} - \eta_j \Bigr| \xrightarrow{P} 0.
\end{align*}
Combining this with the uniform convergence of $\widehat F_n$, $N/n$ and $Y/n$ we can deduce that
\begin{align*}
&\int_0^{\tau_0} w_{nr}^\pi(t)w_{ns}^\pi(t)\Bigl\{ \frac{I\{j=j'\}}{n^{-2}Y_j^\pi(t)Y(t)} - \frac{1}{n^{-2}Y(t)^2} \Bigr\} \,\mathrm{ d }\frac{N}{n}(t)\\
&\to \int_0^{\tau_0} w_r^\pi(t)w^\pi_{r'}(t)\frac{1}{y(t)^2}\Bigl\{ \frac{I\{j=j'\}}{\eta_j} - 1 \Bigr\}  \,\mathrm{ d }\nu(t)
\end{align*}
in probability, where $w_r^\pi(t)=\widetilde w_{r}\{F_0(t)\} (\prod_{j=1}^k\kappa_j){y(t)}$. Moreover, 
\begin{align*}
&\Bigl| \int_{\tau_0}^\infty w_{nr}^\pi(t)w_{ns}^\pi(t)\Bigl\{ \frac{I\{j=j'\}}{n^{-2}Y_j^\pi(t)Y(t)} - \frac{1}{n^{-2}Y(t)^2} \Bigr\} \,\mathrm{ d }\frac{N}{n}(t) \Bigr| \\
&\leq \int_{\tau_0}^\infty \Bigl | \frac{w_{nr}^\pi(t)w_{ns}^\pi(t)}{n^{-2}Y_j^\pi(t)Y(t)}\Bigl\{ I\{j=j'\} - \frac{Y_j^\pi(t)}{Y(t)} \Bigr\}  \Bigr| \mathrm{ d }\frac{N}{n}(t) \\
&\leq K^2 \int_{\tau_0}^{\infty} \frac{N}{n}(t) \to K^2 \int_{\tau_0}^\infty \nu(t) = K^2\sum_{j=1}^k\kappa_j\int_{\tau_0}^\infty \{1-G_j(s)\}\,\mathrm{ d }F_j(s).
\end{align*}
Letting $\tau_0\nearrow \inf\{t\in\R:y(t)=0\}$ gives us
\begin{align}\label{eqn:C_conv_perm}
C^{(rr')}_{njj'} \to \int_0^\infty w_{r}^\pi(t)w_{r'}^\pi(t)\frac{1}{y(t)^2}\Bigl\{ \frac{I\{j=j'\}}{\eta_j} - 1 \Bigr\}  \,\mathrm{ d }\nu(t)=:{\bar \Sigma}^{\pi,(rr')}_{jj'}
\end{align}
in probability. Thus, $ \mbf{\bar Z}^\pi_n$ converges in distribution to a centred multivariate normal distribution with covariance (block-)matrix $\mbf{\bar \Sigma}^\pi = (\mbf{\bar\Sigma}^{\pi,(rr')})_{r,r'=1,\ldots,k}$, where the submatrices are given by the right hand side of \eqref{eqn:C_conv_perm}. Similarly to the argumentation for \eqref{eqn:C_conv_perm}, we can deduce the convergence of the permutation counterpart of the covariance estimator:
\begin{align*}
\widehat \sigma^{2,\pi,(rr')}_j & = I\{j=j'\}\int_0^\infty w_{nr}^\pi(t)w_{nr'}^\pi(t) \frac{n^2}{Y_j^\pi(t)^2}\,\mathrm{ d }\frac{N_j^\pi}{n}(t) \\
&\to I\{j=j'\}\eta_j^{-1}\int_0^\infty w_{r}^\pi(t)w_{r'}^\pi(t) \frac{1}{y(t)^2}\,\mathrm{ d }\nu(t) =: \widehat \sigma^{2,\pi,(rr')}_j. 
\end{align*}
Thus, the permutation counterpart $\mbf{\widehat \Sigma^{\pi}}$ of the covariance matrix estimator converges in probability to the (block-)matrix $\mbf{\Sigma}^\pi= (\mbf{\Sigma}^{\pi,(rr')})_{r,r'}$, where each submatrix $\mbf{\Sigma}^{\pi,(rr')} = \text{diag}(\widehat \sigma^{2,\pi,(rr')}_1,\ldots,\widehat \sigma^{2,\pi,(rr')}_k)$ is a diagonal matrix. The matrix $\mbf{\Sigma}^\pi$ does not coincide with the limiting matrix $\mbf{\bar\Sigma}^\pi$ of the permutation statistic. But the submatrices can be rewritten as  
\begin{align*}
\mbf{\Sigma}^{\pi,(rr')} = \psi_{rr'}\mbf{D}, \quad \mbf{\bar\Sigma}^{\pi,(rr')} = \psi_{rr'}(\mbf{D} - \mbf{1}_k\mbf{1}_k\trans),\quad \mbf{D}=\text{diag}(\eta_1^{-1},\ldots,\eta_k^{-1}),	
\end{align*}
where $\psi_{rr'}=\int_0^\infty w_{r}^\pi(t)w_{r'}^\pi(t)/y(t)^2 \,\mathrm{ d }\nu(t)$. Thus, it is easy to check $\mbf{T}\mbf{\Sigma}^{\pi,(rr')} = \mbf{T}\mbf{\bar\Sigma}^{\pi,(rr')}$ and, consequently, $\mbf{T}^{(m)}\mbf{\Sigma}^\pi = (\mbf{T}\mbf{\Sigma}^{\pi,(rr')})_{r,r'=1,\ldots,m} = \mbf{T}^{(m)}\mbf{\bar\Sigma}^\pi$ follows, which is sufficient for convergence of the joint permutation Wald-type statisticf $S_n^\pi$. Using the same arguments as in the proof of Theorem \ref{theo:null_unc_general}, we can show the regularity of $\mbf{\Sigma}^\pi$. Hence, $(\mbf{T}\widehat{\mbf{\Sigma}^\pi} \mbf{T})^+$ converges in probability to $(\mbf{T}\mbf{\Sigma}^\pi \mbf{T})^+ = (\mbf{T}\mbf{\bar\Sigma}^\pi \mbf{T})^+ $. Finally, Theorem \ref{theo:perm} follows from the continuous mapping theorem.

\section*{Acknowledgement}
Marc Ditzhaus and Markus Pauly were supported by the {Deutsche Forschungsgemeinschaft} (Grant no.  PA-2409 5-1). 
The authors thank Philipp Steinhauer for computational assistance.

\bibliographystyle{chicago}
\bibliography{sample}

\end{document}